\documentclass[]{article}
\usepackage{url}

\usepackage{indentfirst}

\usepackage{algorithm}
\usepackage{algpseudocode}
\usepackage{bm}

\usepackage{amsmath, amsfonts, amsthm}
\usepackage{graphicx}
\usepackage{natbib}
\usepackage[a4paper, total={6in, 8in}]{geometry}

\newcommand{\be}{\begin{equation*}}
\newcommand{\ee}{\end{equation*}}
\newcommand{\bee}{\begin{eqnarray*}}
\newcommand{\eee}{\end{eqnarray*}}

\newtheorem{thm}{Theorem}
\newtheorem{cor}[thm]{Corollary}

\begin{document}

\title{Fidelity of Hyperbolic Space for Bayesian Phylogenetic Inference}
% Each important word in the title should begin with a capital letter

% List of authors, with corresponding author marked by asterisk
\author{Matthew Macaulay$^{1,\ast}$, Aaron E. Darling$^{1,2}$, and
Mathieu Fourment,$^{1}$\\[4pt]
\textit{$^{1}$~University of Technology Sydney, Australian Institute for Microbiology \& Infection,}\\ \textit{Ultimo NSW 2007, Australia}
\textit{$^{2}$~Illumina Australia Pty Ltd, Ultimo NSW 2007, Australia}
\\[2pt]
\textit{*matthew.macaulay@uts.edu.au}}
% Identify the name, address, telephone/fax numbers, and e-mail address for the author who will receive proofs and be designated the "corresponding author" in text.

\markboth%
% First field is the short list of authors
{Macaulay, Darling and Fourment}
% Second field is the short title of the paper
{Hyperbolic Bayesian Phylogenetics}
% This should be shortened version of the title and no greater than 50 characters

\maketitle

\begin{abstract}
{
Bayesian inference for phylogenetics is a gold standard for computing distributions of phylogenies.
It faces the challenging problem of moving throughout the high-dimensional space of trees.
However, hyperbolic space offers a low dimensional representation of tree-like data.
In this paper, we embed genomic sequences into hyperbolic space and perform hyperbolic Markov Chain Monte Carlo for Bayesian inference.
The posterior probability is computed by decoding a neighbour joining tree from proposed embedding locations.
We empirically demonstrate the fidelity of this method on eight data sets.
The sampled posterior distribution recovers the splits and branch lengths to a high degree.
We investigated the effects of curvature and embedding dimension on the Markov Chain's performance.
Finally, we discuss the prospects for adapting this method to navigate tree space with gradients.}
{Bayesian inference, phylogenetics, hyperbolic embeddings, Markov Chain.}
\end{abstract}

Bayesian phylogenetics seeks the posterior distribution over discrete tree topologies and continuous tree branch lengths and evolutionary model parameters given an alignment of nucleotide sequences.
Computing the posterior is analytically intractable so it is approximated using Markov chain Monte Carlo (MCMC)~\citep{yang1997bayesian, larget1999markov}.
It is a workhorse algorithm of Bayesian phylogenetics that draws samples from the posterior distribution by proposing new samples that are ``nearby'' the current state of the Markov chain.
However, the space of phylogenetic trees is super-exponential in the number of topologies and navigating tree space is difficult~\citep{whidden2015quantifying, harrington2021properties}.
Most methods to propose a new tree topology are too simplistic, e.g. nearest neighbour interchange (NNI) and subtree prune and re-graft (SPR), and improvements on these by guided proposals comes at the cost of additional computation~\citep{hohna2012guided}.
This all poses the dilemma of how to best represent and navigate the space of phylogenies.

Hyperbolic space offers quality embeddings of tree-like data in low dimensions.
Hyperbolic representations have successfully clustered data hierarchically into trees~\citep{NEURIPS2020_ac10ec1a, gu2018learningb, monath2019gradientbased}.
They optimise a carefully devised objective function to find an optimal embedding.
Unlike those methods, the objective function for Bayesian phylogenetics is prescribed by the model of evolution and prior distribution.
Nonetheless, embedding genomic sequences as points in hyperbolic space and working in the embedding space could provide a way to move through tree space with an improved notion of``locality'' --- small changes to the embedding locations of sequences produce small changes in their distances on the tree.
The embedding may enable a more natural representation of sequence divergences, potentially enabling a MCMC sample to make more complex changes to the model (topology, branch lengths etc) in a single move
Recent works reviewed by~\citet{Iuchi2021representation} have used representations, or specifically hyperbolic representations to learn a maximum likelihood phylogeny~\citep{wilson2021learning, matsumoto2021novel}.
However, their Bayesian counterpart is missing.

Bayesian phylogenetic practitioners increasingly need scalable methods to deal with larger sets of sequences, as highlighted by the current global pandemic of SARS-CoV-2~\citep{ki2022variational}.
A demonstration that MCMC works in hyperbolic space could open up a wealth of possibilities for working with embeddings for Bayesian phylogenetics.
As opposed to machine learning techniques, our work decodes a tree and uses the phylogenetic posterior probability as a cost function for hyperbolic embeddings.
It also generates a distribution of phylogenetic trees rather than a point estimate.

To assess whether hyperbolic embeddings can represent a posterior distribution of trees, we perform MCMC on trees decoded from an embedding of nucleotide sequences.
Trees are decoded using the neighbour-joining (NJ) algorithm~\citep{saitou1987neighborjoining} before computing their likelihood and prior probability for each MCMC generation.
We implement this MCMC in a python package called Dodonaphy, which is available on GitHub at \url{https://github.com/mattapow/dodonaphy}.

The goal of this paper is to empirically demonstrate the fidelity of hyperbolic embeddings for Bayesian phylogenetics.
We begin by detailing the concepts from Bayesian phylogenetics, hyperbolic embeddings of trees and MCMC that are necessary to present the method.
Once devised, we investigate the fidelity of the method and illustrate the embedding landscape.
We then quantify the effects of the curvature and dimension of the embedding space on the MCMC performance.
We also explore how a proposal distribution in the embedding space transfers to tree space.
Our last result is the algorithm complexity and run-times.
Finally, we discuss the numerous future research possibilities that this method opens.

\bigskip
% Each important word in the heading level 1 should begin with a capital letter; no heading for the introduction
% Each important word in the heading level 2 should begin with a capital letter
% First word and proper nouns only should begin with a capital letter in heading level 3
\section{Materials and Methods}
\label{sec:method}

\subsection{Bayesian Phylogenetics}
In phylogenetics, the posterior probability is the probability of a tree and evolutionary model parameters given an alignment of nucleotide sequences $\psi$.
We refer to a phylogenetic tree $T$ which includes both the topology and branch lengths of the tree.
The posterior probability of a tree is
\be
p(T | \psi) = \dfrac{p(\psi | T) p(T)}{p(\psi)}.
\ee
Dodonaphy uses a simple model of evolution, the Jukes-Cantor (JC69) model~\citep{jukes1969evolution} to compute the likelihood given a tree $p(\psi | T)$ using the Felsenstein pruning algorithm~\citep{felsenstein1973maximum}.

For the prior $p(T)$, we utilise a uniform distribution over tree topologies and a Gamma-Dirichlet prior probability of a tree's branch lengths, as previously suggested by~\citet{rannala2012tail}.
The Gamma-Dirichlet prior assigns a Gamma distribution $\Gamma ( \alpha, \beta)$ with shape $\alpha=1$ and rate $\beta = 0.1$ on the total tree length (sum of branch lengths), before dividing this into individual branches with a Dirichlet distribution $\text{Dir}(\alpha_{k}=1.0)$.

The state-of-art phylogenetic software MrBayes~\citep{ronquist2003mrbayes} offers both of these models for the prior and likelihood, which allows us to directly compare its results with Dodonaphy.
We use two \textit{golden} runs of MrBayes with $10^{9}$ iterations and consider these the \textit{ground truth} posterior distribution.

We approximate the posterior distribution through MCMC sampling, which does not require computing the intractable probability of the data $p(\psi)$.
As detailed later, it only computes the likelihood and prior probabilities of a tree.
However, it does require generating proposals, for which, we move to a hyperbolic embedding.

\subsection{The Hyperboloid Model}
One common model of Hyperbolic space is the upper sheet of a hyperboloid $\mathbb{H}^{d} = \{\bm{x} \in \mathbb{R}^{d+1}: \langle \bm{x}, \bm{x} \rangle = -1 \}$ where the Lorentz inner product is defined as
\begin{gather*}
\langle \bm{x}, \bm{y} \rangle = \bm{x}^{\mathsf T} H \bm{y}, \qquad
H = 
\begin{bmatrix}
-1 & 0\\
0 & I
\end{bmatrix}
\in \mathbb{R}^{d+1} \times \mathbb{R}^{d+1}.
\end{gather*}
The length of the geodesic between two points on the sheet is given by the metric
\be \label{eq:metric}
d(\bm{x}, \bm{y}) = \dfrac{1}{\sqrt{-\kappa}} \text{arcosh}(-\langle \bm{x}, \bm{y} \rangle).
\ee
In this model, for a fixed distance $d(\bm{x}, \bm{y})$, a more negative curvature $\kappa<0$ stretches how far apart the points are on the curved sheet, imparting more curvature between the points.
Conversely, to maintain a fixed distance as $\kappa \to 0$, points move towards the origin where the space is flatter.

The first coordinate is determined by the last $d$ coordinates according to $\langle \bm{x}, \bm{x} \rangle = -1$ to stay on the sheet.
This rearranges to give
\be
x_{0} = \sqrt{x_{1}^{2}+x_{2}^{2}+... x_{d}^{2}+ 1}.
\ee
Thus we need only store the last $d$ points in $\mathbb{R}^{d}$ and project up exactly onto the hyperboloid by $\phi: \mathbb{R}^{d} \to \mathbb{H}^{d} \subset \mathbb{R}^{d+1}$. Explicitly, we use:
\begin{equation} \label{eq:x0}
\phi(\bm{x}) =
\begin{bmatrix}
x_{0} \\
\bm{x}
\end{bmatrix}
\end{equation}
projecting the first coordinate and leave the rest unchanged.
The inverse mapping omits the first coordinate: $\phi^{-1}([x_{0}, x_{1}, ..., x_{d}]^{\mathsf T}) = [x_{1}, ..., x_{d}]^{\mathsf T}$.

\subsection{Hyperbolic Embeddings}
To initialise an embedding of the taxa, Dodonaphy begins by minimising the stress of the embedding to a distance matrix.
The distance matrix may either derive from a provided starting tree or the genetic distances between the aligned sequences.
The stress of an embedding of points $\bm{x}_{i} \in \mathbb{H}^{d}$ to a given set of pairwise distances $D_{ij}$ is 
\be
\sigma^{2} = \sum_{i,j} (D_{ij} - d(\bm{x}_{i}, \bm{x}_{j}))^{2}.
\ee
Minimising this stress is a challenging non-convex problem, however, Hydra+ is a recent algorithm that provides a fast approximate solution~\citep{keller-ressel2020hydra}.
It begins at a solution to the much simpler corresponding strain minimisation problem:
\be
\epsilon^{2} = \sum_{i, j}  (\text{cosh}(\sqrt{\kappa} D_{ij}) - \langle \bm{x}_{i}, \bm{x}_{j} \rangle )^{2},
\ee
finding a set of locations $x$ to minimise $\epsilon$ though an Eigen-decomposition.
Then it minimises the stress through gradient-based optimisation.
We reimplemented Hydra+ as a Python package available on GitHub at \url{https://github.com/mattapow/hydraPlus}.
%AD 220604: It's good practice to deposit software into zenodo around the time of manuscript publication. zenodo has github integration so the process is pretty easy via the zenodo website. github is a private company and does not give the guarantees on long term archival that academic librarians would typically expect. Maybe you could even think about a 1 page paper on your hydra+ code in something like the Journal of Open Source Software (JOSS) as a little CV booster?

We embed the $n$ labelled taxa as a set of points in a hyperbolic space $X = \{\bm{x}_{i}\}$, $i = 1, 2, ..., n$ using hydra+.
By representing taxa as points on the hyperboloid $\bm{x}_{i} \in \mathbb{H}^{d}$ we can work directly with their hyperbolic distances $d(\bm{x}_{i}, \bm{x}_{j})$.

\subsection{Tree Decoding}
Dodonaphy then employs neighbour-joining to decode a tree $T = \text{NJ}(X)$ from a set of embedding locations $X$.
This decodes both branch lengths and a tree topology at once.
The NJ algorithm takes pairwise distances as input, which is efficiently computed in the hyperboloid model of hyperbolic space using linear algebra~\citep{chowdhary2018improved}.
The main advantage of using NJ is its consistency: it correctly constructs a tree when the given distances fit on that tree.
However, instead of the sequence distances, Dodonaphy utilises the hyperbolic distances between taxa for NJ.
This allows Dodonaphy to move through tree space by moving the set of $n$ tip vectors $X$ through hyperbolic space $\mathbb{H}^{n\times d}$.

\subsubsection{Unique tree decoding}
The neighbour-joining algorithm fits a unique unrooted tree provided that the pairwise distances satisfy the four-point condition: $\forall \bm{w}, \bm{x}, \bm{y}, \bm{z}$
\be
d(\bm{x}, \bm{w}) + d(\bm{y}, \bm{z}) \leq \max\{d(\bm{x}, \bm{y}) + d(\bm{z}, \bm{w}), d(\bm{x}, \bm{z}) + d(\bm{y}, \bm{w})\}.
\ee
When this occurs for all pairs of points, the distances are said to be additive and a tree consistent with the given distances can be recovered by NJ.

Hyperbolic space is often chosen for embedding trees because there exists a bound $\delta > 0$ on how much hyperbolic distances violate the four-point condition:
\begin{equation} \label{eq:delta}
d(\bm{x}, \bm{w}) + d(\bm{y}, \bm{z}) \leq \max\{d(\bm{x}, \bm{y}) + d(\bm{z}, \bm{w}), d(\bm{x}, \bm{z}) + d(\bm{y}, \bm{w})\} + 2\delta.
\end{equation}
Such spaces are called $\delta$-hyperbolic.
Furthermore, this bound $\delta$ can be scaled to become arbitrarily small by making the curvature $\kappa$ more negative~\citep{wilson2021learning}.

\subsection{MCMC}
Our interest is forming posterior distributions of trees, which we achieve through MCMC sampling of the posterior distribution.
MCMC is a classical method for computing the posterior distribution $p(T | \psi)$ given some data $\psi$.
Starting at an initial location in the posterior $X$, a new point $X^{*}$ is proposed from a proposal distribution $g(X^{*}|X)$ and accepted or rejected according to the Metropolis-Hastings algorithm.
A new proposal state $X^{*}$ decodes a tree $T^{*} = \text{NJ}(X^{*})$ which is accepted with probability $\min(1, \alpha)$, where the acceptance ratio is
\be
\alpha = \frac{p(\psi | T^{*}) p(T^{*})}{p(\psi | T) p(T)} \frac{g(X|X^{*})}{g(X^{*}|X)}
\ee
Importantly, it does not depend on the intractable component of the posterior $p(\psi)$.
Choosing a symmetric proposal $g(X^{*}|X)$, such as a Normal distribution, makes the last term simply cancel out.
After a period of burn-in, points generated from this simple procedure compose samples from the targeted posterior distribution.

\subsubsection{Improved MCMC}
We make use of two modifications to the original MCMC algorithm.
First, to improve mixing, Metropolis Coupled MCMC can use multiple chains that can swap intermittently through an MCMC move~\citep{geyer1991markov, altekar2004parallel}.
The advantage of this is the ability to heat chains to move out of local optima more freely and it is routinely used in MrBayes.
This is achieved by raising the acceptance ratio to a power given by a temperature parameter $\alpha^{\tau}$, $0\leq \tau \leq 1$.
We use four chains with temperatures $\tau_{i} = 1 / (1 + \lambda i)$, where $i=0, 1, 2, 3$ is the chain index and $\lambda=0.1$ is fixed, as is done in MrBayes.
This ensures that the first chain remains \textit{cold} ($\tau_{0}=1$) and correctly samples from the posterior distribution.

The second idea is to adapt the covariance matrix of proposals using the Robust Adaptive Metropolis (RAM) algorithm~\citep{vihola2012robust}.
The RAM algorithm is a variant of the Metropolis-Hastings algorithm that uses a covariance matrix to adapt the proposal distribution.
During a warm-up phase, we tune the initially diagonal covariance matrix simply by scaling it to achieve the target acceptance rate of $0.234$, see appendix~A.
After the warm-up phase up we employ the RAM algorithm for the duration of the MCMC.
Simultaneously, it tunes the acceptance rate to ensure computational efficiency.

\subsubsection{MCMC Proposals}
To propose a new MCMC point we draw from a multivariate Gaussian centred over the current taxa positions.
We centre the Gaussian at the locations of each taxa $\bm{x}_{i} \in \mathbb{H}^{d}$, by omitting the first (prescribed) dimension and concatenating all locations into a single vector $\vec{X} \in \mathbb{R}^{n \times d}$. Note the distinction from the set of points denoted $X = \{\bm{x}_{i}\}$ with $i = 1, 2,..., n$ indexing the taxa.
The covariance matrix is initialised to a scalar multiple of the identity $\Sigma = \zeta I_{n\times d}$, with $\zeta=0.1$, before being tuned by the MCMC algorithms.
Explicitly, each proposal point is drawn as $X^{*} \sim \mathcal{N}(\vec{X}, \Sigma)$.
The proposal points are then projected onto the hyperboloid exactly using Eq.(\ref{eq:x0}).

\subsection{Experiments}
\begin{figure}[htbp]
\begin{center}
\includegraphics[width=.7\textwidth]{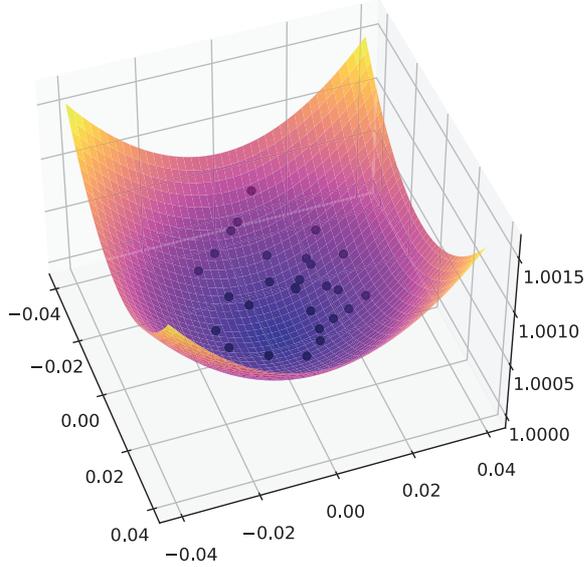}
\caption{Hyperbolic Embedding of an MCMC sample of DS2. 29 points (taxa) lie on a two dimensional hyperboloid sheet $\mathbb{H}^{2}$ in $\mathbb{R}^{3}$.}
\label{fig:visu}
\end{center}
\end{figure}
In summary, Dodonaphy embeds pairwise tip distances $D$ with hydra+ and performs phylogenetic MCMC on a hyperboloid, algorithm~\ref{alg:dodo}.
An example embedding is illustrated in figure~\ref{fig:visu}.
The MCMC is run with a Gamma-Dirichlet prior on the branch lengths and a tree likelihood under a JC69 model of evolution.
The MCMC runs four parallel chains and performs ten MCMC chain swap moves every $10^{3}$ generations.
MCMC proposals are multivariate normal distributions for all tip vectors.
An initial warm-up period of $10^{4}$ iterations serves to tune the covariance matrix before switching to the RAM algorithm until reaching $10^{6}$ iterations.
Dodonaphy draws $10^{4}$ evenly spaced tree samples throughout the simulation.
Unless otherwise stated, simulations are run in three dimensions with curvature $\kappa=-1$.

\begin{algorithm}
\label{alg:dodo}
\caption{Hyperbolic MCMC algorithm in Dodonaphy.}
\begin{algorithmic}[1]
\Procedure{Dodonaphy}{$\psi, \kappa, d, \text{iterations}$}
	\State $D \gets \text{genetic distances}(\psi)$ \Comment{Pair-wise alignment distances}
	\State $X \in \mathbb{H}^{n\times d} \gets \text{hydra+}(D, \kappa, d)$ \Comment{Embed distances}
	\State $T \gets \text{NJ}(X)$ \Comment{Decode NJ tree}
	\For{\text{iterations}}
		\State $X^{*} \gets \text{draw from } \phi(\mathcal{N}(\vec{X}, \Sigma))$ \Comment{Propose new state}
		\State $T^{*} \gets \text{NJ}(X^{*})$  \Comment{Decode NJ tree}
		\State $X, T \gets X^{*}, T^{*} \text{ with probability } \alpha$ \Comment{Metropolis step}
	\EndFor 
\EndProcedure
\end{algorithmic}
\end{algorithm}

%\subsubsection{Bayesian Phlogenetics Via Hyperbolic Embeddings}

We looked at eight real datasets (DS1-DS8) compiled in \citet{lakner2008efficiency} to evaluate the fidelity of hyperbolic space for Bayesian phylogenetics.
The datasets contain between $n=27$ and $n=64$ taxa of aligned nucleotide sequences and are commonly used for phylogenetic benchmarking~\citep{whidden2020systematic}.
They are freely available from treebase~\citep{vos2012nexml} with their identifier listed in table~\ref{tab:data}.
The sequence lengths vary from $378$ to $2520$.
We compress them into $256 \leq L \leq 1252$ unique site patterns by weighting each pattern.

\begin{table*}[htp]
\caption{Data sets used in this analysis.}
\begin{center}
\begin{tabular}{lllll}
\hline
Label & n & sites & Type of data & Treebase \\ \hline
DS1 & 27 & 1949 & rRNA, 18s & M2017 \\
DS2 & 29 & 2520 & rDNA, 18s & M2131 \\
DS3 & 36 & 1812 & mtDNA, COII (1–678), cytb (679–1812) & M127 \\
DS4 & 41 & 1137 & rDNA, 18s & M487 \\
DS5 & 50 & 378 & Nuclear protein coding, wingless & M2907 \\
DS6 & 50 & 1133 & rDNA, 18s & M220 \\
DS7 & 59 & 1824 & mtDNA, COII and cytb & M2449 \\
DS8 & 64 & 1008 &  rDNA; 28s & M2261\\ \hline
\end{tabular}
\end{center}
\label{tab:data}
\end{table*}

\bigskip
\section{Results}
\label{sec:results}
To compare posterior distributions of phylogenetic trees we consider the distributions of splits, branch lengths, total tree lengths and the posterior tree probability.
First, we present these in detail for DS1.
Then we illustrate Dodonaphy's capacity over a range of datasets, curvatures and embedding dimensions.

To summarise the frequency of splits appearing in the posterior, we employ the average standard deviation of split frequencies (ASDSF)~\citep{lakner2008efficiency}.
It is a common statistic used to compare the splits between tree distributions, for example as used by MrBayes.
A threshold of $\text{ASDSF} < 0.05$ is commonly used to indicate that the split frequencies from two distributions closely match each other.

\subsection{Embedding Fidelity}
We test Dodonaphy's capacity for MCMC starting directly from sequence data.
Dodonaphy constructs the neighbour-joining tree directly from the multiple sequence alignment.
The pairwise evolutionary distances (number of substitutions per site) between sequences are fed into hydra+ to embed the taxa as hyperbolic points.
In this section, Dodonaphy runs for $10^{7}$ generations.

\begin{figure}[htbp]
\begin{center}
\includegraphics[width=.33\linewidth]{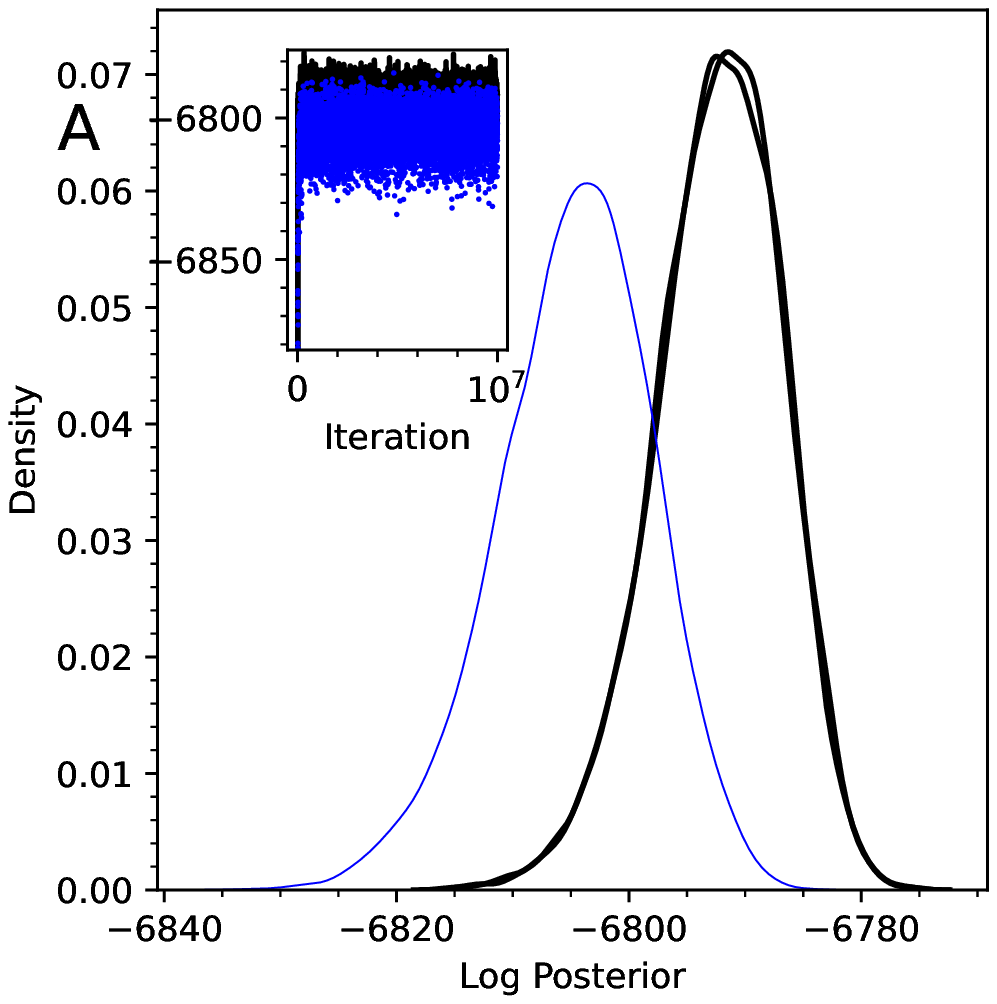}%
\includegraphics[width=.33\linewidth]{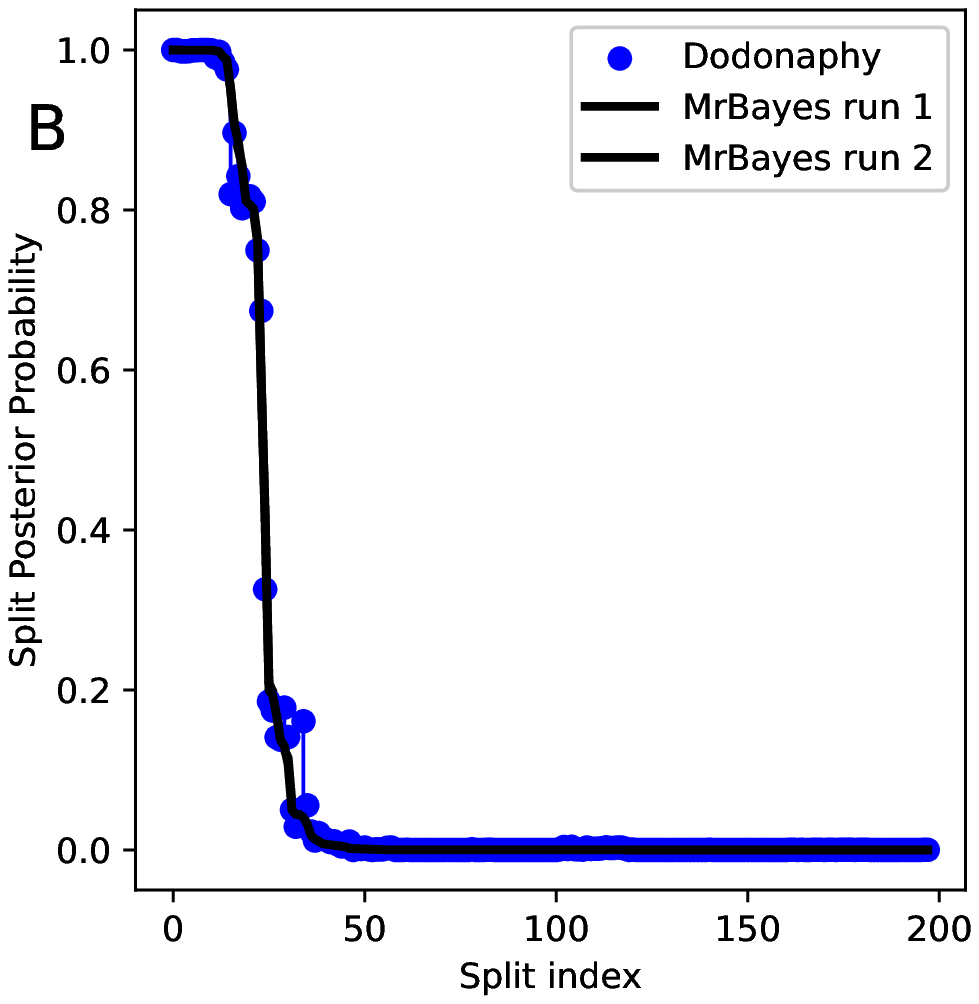}
\includegraphics[width=.33\linewidth]{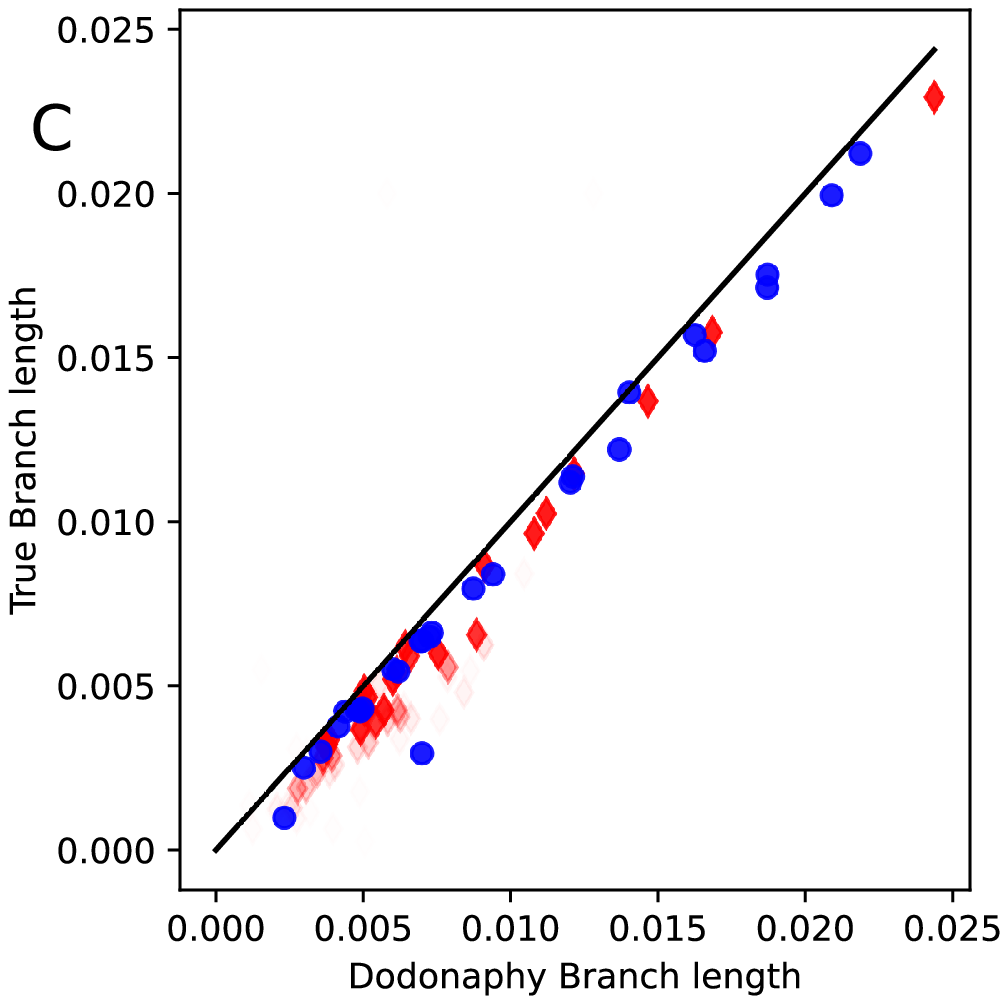}%
\includegraphics[width=.33\linewidth]{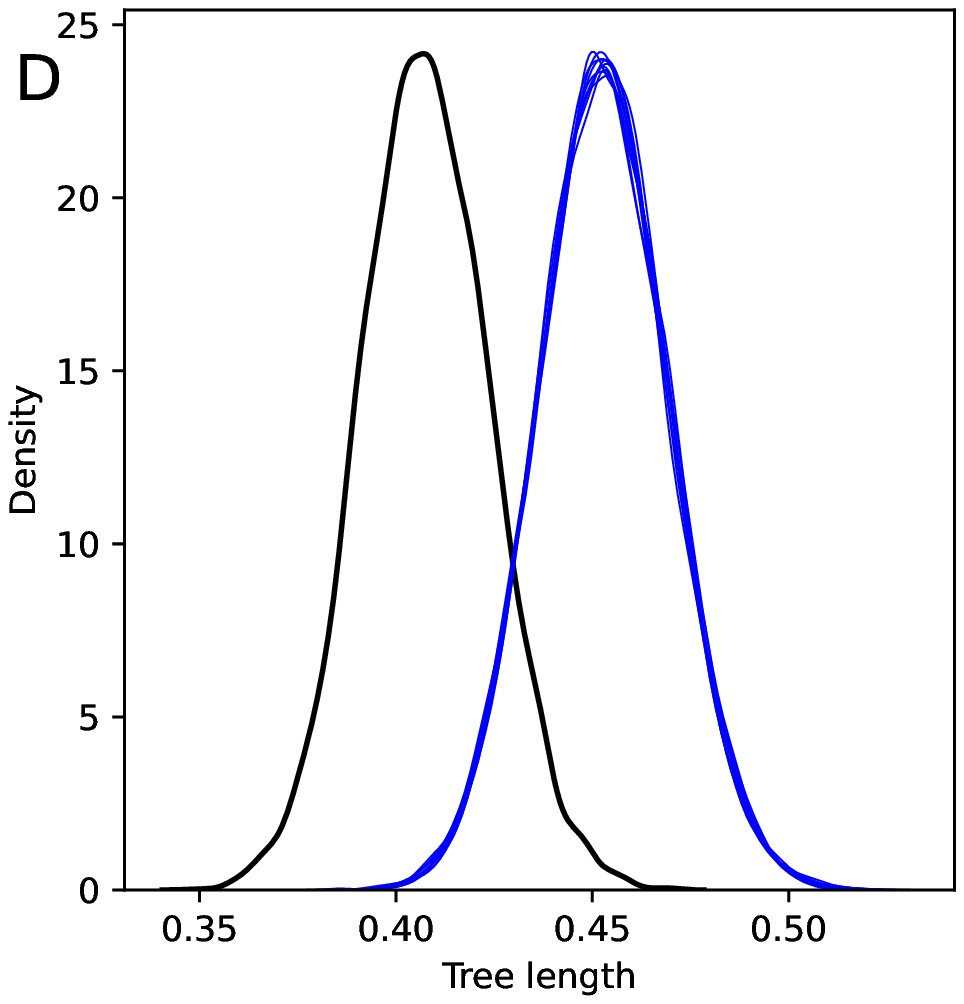}
\caption{Comparison between MrBayes and Dodonaphy's MCMC starting from the evolutionary distances. Comparison of
(a) posterior probability trace,
(b) split frequencies,
(c) mean branch lengths (leaf edges in blue circles, internal edges in red diamonds),
(d) total tree length estimation of 10 repeats.
Markers in (c) are shaded by the frequency of appearance in the golden run.}
\label{fig:start_msa}
\end{center}
\end{figure}

Figure~\ref{fig:start_msa}a) demonstrates that Dodonaphy can move from this starting state into a similar region of tree space as the golden runs of MrBayes.
The starting probability for the joint distribution of the data and parameters for Dodonpahy is about $-8266$, whereas for MrBayes it initialises to $-9822$.
%MF: you could calculate the ESS to back this up
%MM: I'm not sure how to compute this.
The trace plot illustrates that the burn-in period appears completed and the MCMC does not get stuck in one state for too long, pointing to good mixing.

Dodonaphy explores the posterior splits with considerable fidelity, figure.~\ref{fig:start_msa}b).
It recovered every split appearing in the golden run with support above $10^{-3}$.
Of the $64$ splits from MrBayes that Dodonaphy misses, only one split was visited by either golden run more than once.
Indeed, the ASDSF between Dodonaphy and the first golden run is $0.003$, well below the threshold for equivalence.
%AD 220604: it's not important but I wonder if that one split that was found more than once was from a nearby iteration, i.e. perhaps the mcmc didn't get many chances to try to move away from it? 

The mean length of splits appearing in both the golden runs and Dodonaphy are compared in figure~\ref{fig:start_msa}c).
They generally match well, however, Dodonaphy tends to slightly overestimate branch lengths.
Consequently, the mean tree length posterior is overestimated by Dodonaphy:  $0.453$ compared to $0.408$, a ratio of $1.11$.
The mean variance of the total tree length matches to a high degree: a ratio of $2.830 \times 10^{-4} / 2.656 \times 10^{-4} = 1.007$.
Compared to the golden run, this overestimation reduces the overall log posterior, figure~\ref{fig:start_msa}a) inset.

Repeated runs of Dodonpahy yield similar results on the tree length, figure~\ref{fig:start_msa}d).
The support, mean location and shape of the tree length distribution are self-consistent.

For the following sections, the MCMC chains are initialised from the consensus tree from MrBayes and run for a shorter $10^{6}$ generations.

\subsection{Embedding Curvature}
The embedding curvature is a freely chosen hyper-parameter of the embedding space.
Figure~\ref{fig:crv} reveals how a wide range of curvatures are suitable for all datasets.
Panel a) shows that, compared to a golden run, the ASDSF falls below the threshold of $0.05$ for curvatures in the range $-100\leq\kappa\leq-1$.
Similarly, the relative difference in the median total posterior tree length $\hat{l}$ to a golden run $\hat{l}_{mb}$ falls within $10\%$ for $-100\leq\kappa\leq-1$ range for all datasets.

\begin{figure}[htbp]
\begin{center}
\includegraphics[width=.33\linewidth]{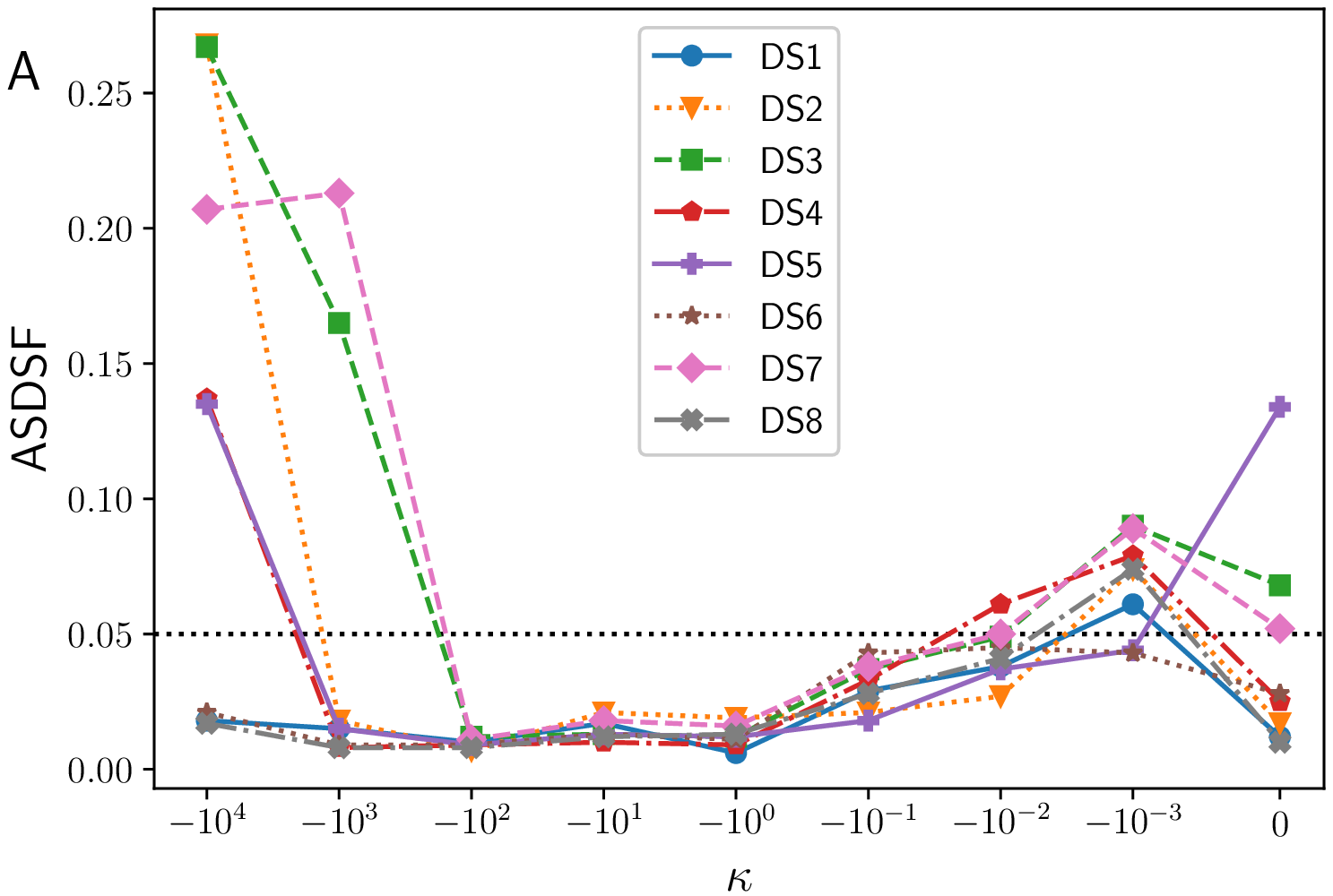}%
\includegraphics[width=.33\linewidth]{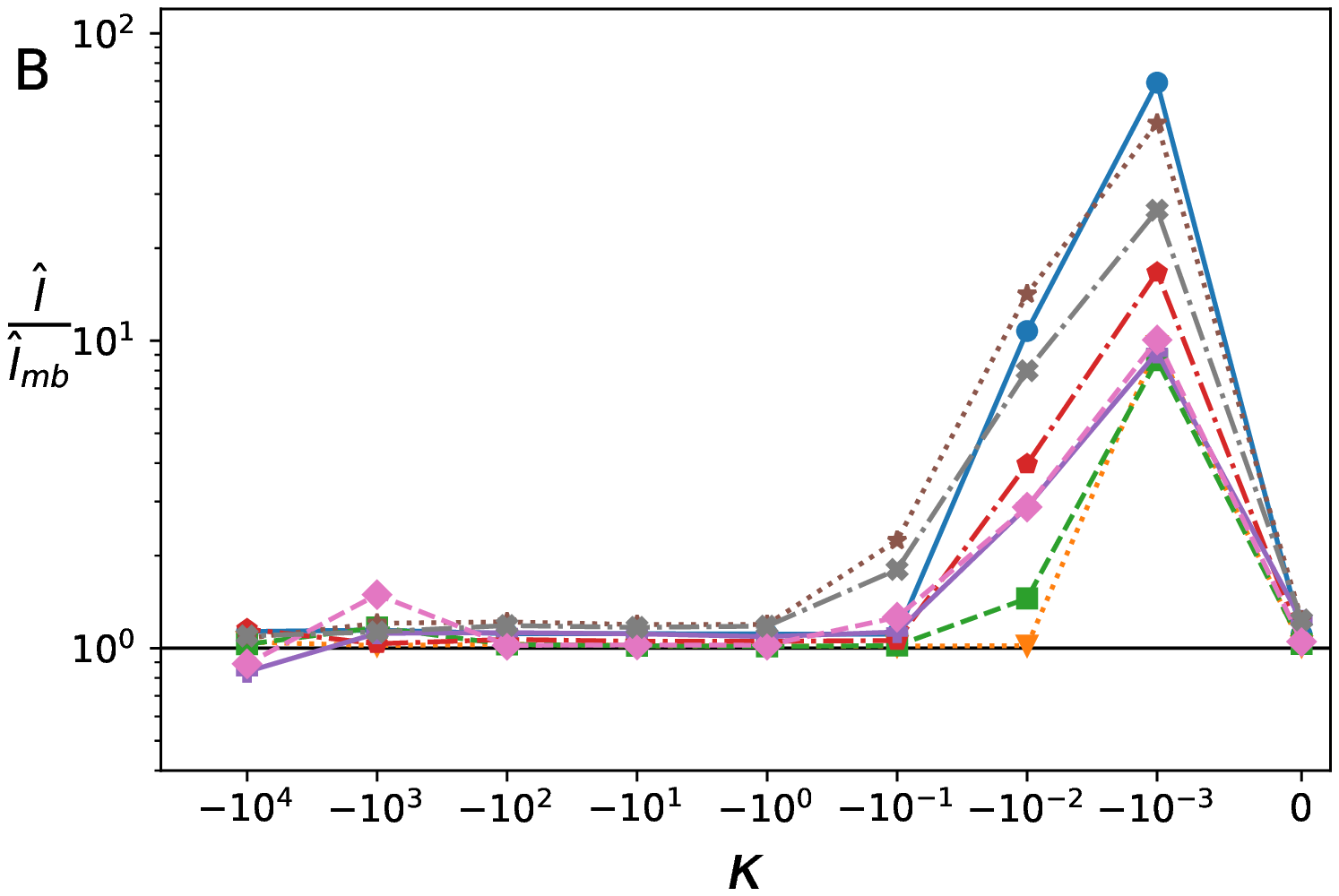}%
\includegraphics[width=.33\linewidth]{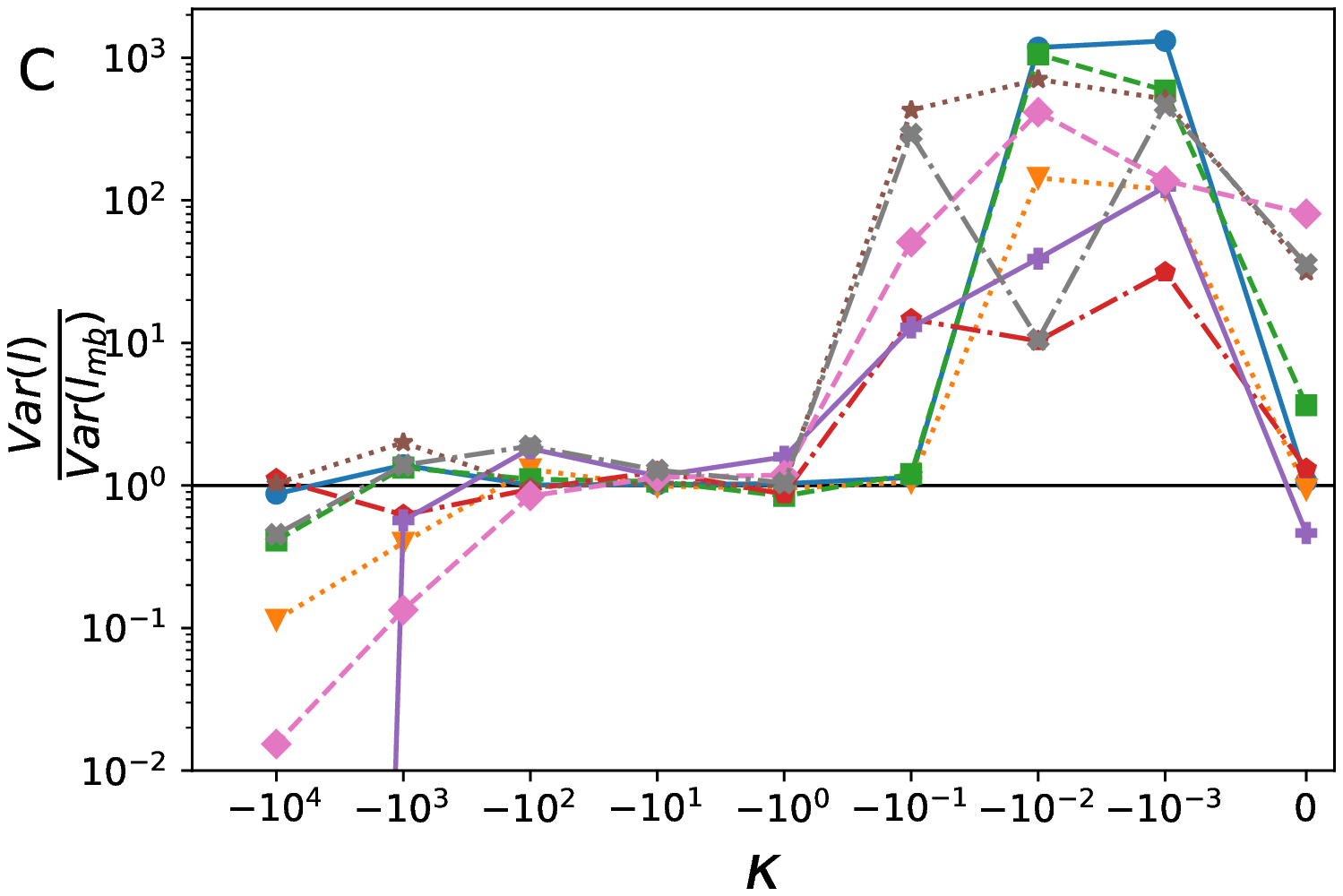}
\caption{Effect of embedding curvature on posterior distribution.
Comparison to true posterior of:
a) ASDSF,
b) relative difference in median tree length,
c) relative difference in variance of tree length.
The truncated variance ratio for DS5 is approximately zero.
}
\label{fig:crv}
\end{center}
\end{figure}

As the space becomes flatter ($\kappa \to 0$), the space becomes more Euclidean.
For $\kappa=0$, we replace the hyperbolic distance with the regular Euclidean metric.
These embeddings were initialised by feeding hydra+ a minuscule curvature of $10^{-10}$.
Note that MCMC in Euclidean space differs from vanilla MCMC because the tips are still embedded.
As the space becomes flatter, the variance of tree lengths is generally larger and the trees tend to be longer, but not monotonically.
Euclidean embeddings had better tree length estimates than weakly hyperbolic embeddings.
However, the ASDSF significantly worsens in this limit, indicating that the wrong trees were recovered in Euclidean space.

In the other extreme, when the curvature is below $-100$, the variance of tree lengths becomes smaller.
MCMCs with low variance are consistent with being constrained to a local optimum in tree space.
This is also evidenced by worsening ASDSFs in this limit, signifying that, in same data sets, the tree topologies of the full posterior are not fully explored.

\subsection{Tree Length Continuity}
The $\delta$-hyperbolic property of hyperbolic space makes it suitable for distance-based phylogenetics.
We build on work by~\citet{wilson2021learning} to highlight the continuity of the decoded tree lengths.
First, let $\delta_{d, \kappa}$ denote the minimal such error $\delta$ in the four point condition in $d$ dimensions with curvature $\kappa$, Eq.~\ref{eq:delta}.

\begin{thm} \citep{wilson2021learning}
For any $d\geq 2$ and $\kappa<0$, there exists $\delta>0$ such that $\mathbb{H}^{d}$ with curvature $\kappa$ is $\delta$-hyperbolic. Furthermore,
\be
\delta_{d, \kappa} = \frac{1}{\sqrt{-\kappa}} \delta_{d, 1}.
\ee
\end{thm}

This scaling is significant for neighbour-joining because it leads to additive embedding distances.
A distance matrix $D$ is \textit{additive} if there is a weighted tree such that the distances on the tree $D_{T}$ match the given distances $D = D_{T}$.
The distances are additive when the four point condition is satisfied $\delta = 0$.
When the input distances are additive, neighbour-joining decodes a consistent tree.
That is, it decodes a unique tree with correct tree distances: $D = D_{T}$.

Neighbour joining remains consistent when the distances are almost additive.
For the following theorem, we need the $l_{\infty}$ norm, written $||D-D_{T}||_{\infty}$, which is the maximum element-wise difference between $D$ and $D_{T}$.
The radius of a method $f$ is the maximum $\alpha$ in
\be
||D-D_{T}||_{\infty} < \alpha \min_{e\in E(T)}(l_{e})
\ee
so that $f(D) = T$ produces the same tree with edges $E(T)$.

\begin{thm} \citep{atteson1999performance} 
\label{thm:atterson}
The $\l_{\infty}$ radius of neighbour joining is $\frac{1}{2}$.
\end{thm}

\begin{cor}
\label{cor:length_continuous}
In the limit $\kappa \to -\infty$, the length of a decoded tree $l(\text{NJ}(X))$ is continuous in the embedding locations $X$.
\end{cor}
\begin{proof}
In the limit $\kappa \to -\infty$, the error becomes arbitrarily small $\delta_{d, \kappa} \to 0$.
In particular, there exists $\kappa < 0$ such that $||D-D_{T}||_{\infty} < \alpha \min_{e\in E(T)}(l_{e})$.
When this occurs, theorem~\ref{thm:atterson} ensures the correct topology is decoded and the distances are arbitrarily close to \textit{additive}.
Consequently, the hyperbolic distances match the tip-tip distances on the tree, even as the topology changes.
Since the hyperbolic distances are continuous in the embedding locations, so are the tree lengths.
\end{proof}

If $\kappa$ is too large, NJ may not reconstruct a unique tree, making the MCMC dependent on the particular taxa selected to merge in the algorithm.
This difficulty could partly explain the large variances observed in tree lengths when $\kappa > -1$.

In the next section, we explore how the continuity of decoded distances translates into the posterior landscape.

\subsection{Posterior Landscape}
To gain an intuition for the effect of curvature on the posterior landscape, we mapped it for one taxon's position.
We selected a set of embedding locations from an MCMC run in two dimensions $\mathbb{H}^{2}$.
We then fixed all taxa locations except one, for which we performed a grid search.
At each embedding location of this node, an NJ tree is decoded and the joint probability of this tree and the data is recorded.
To plot in two dimensions, we projected the locations on the hyperboloid onto $\mathbb{R}^{2}$ by $\phi^{-1}$.

\begin{figure}[htbp]
\begin{center}
\includegraphics[width=.33\linewidth]{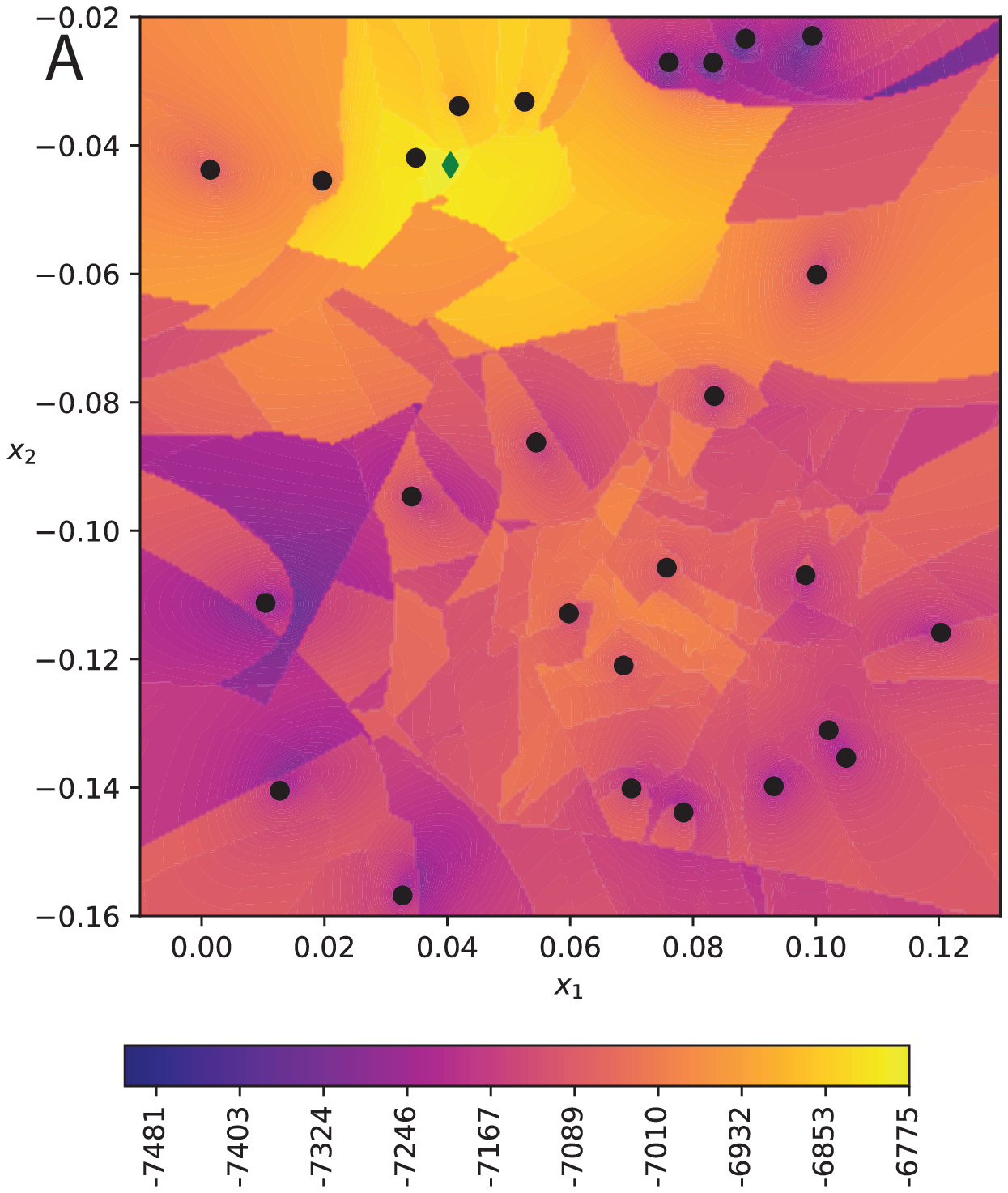}%
\includegraphics[width=.33\linewidth]{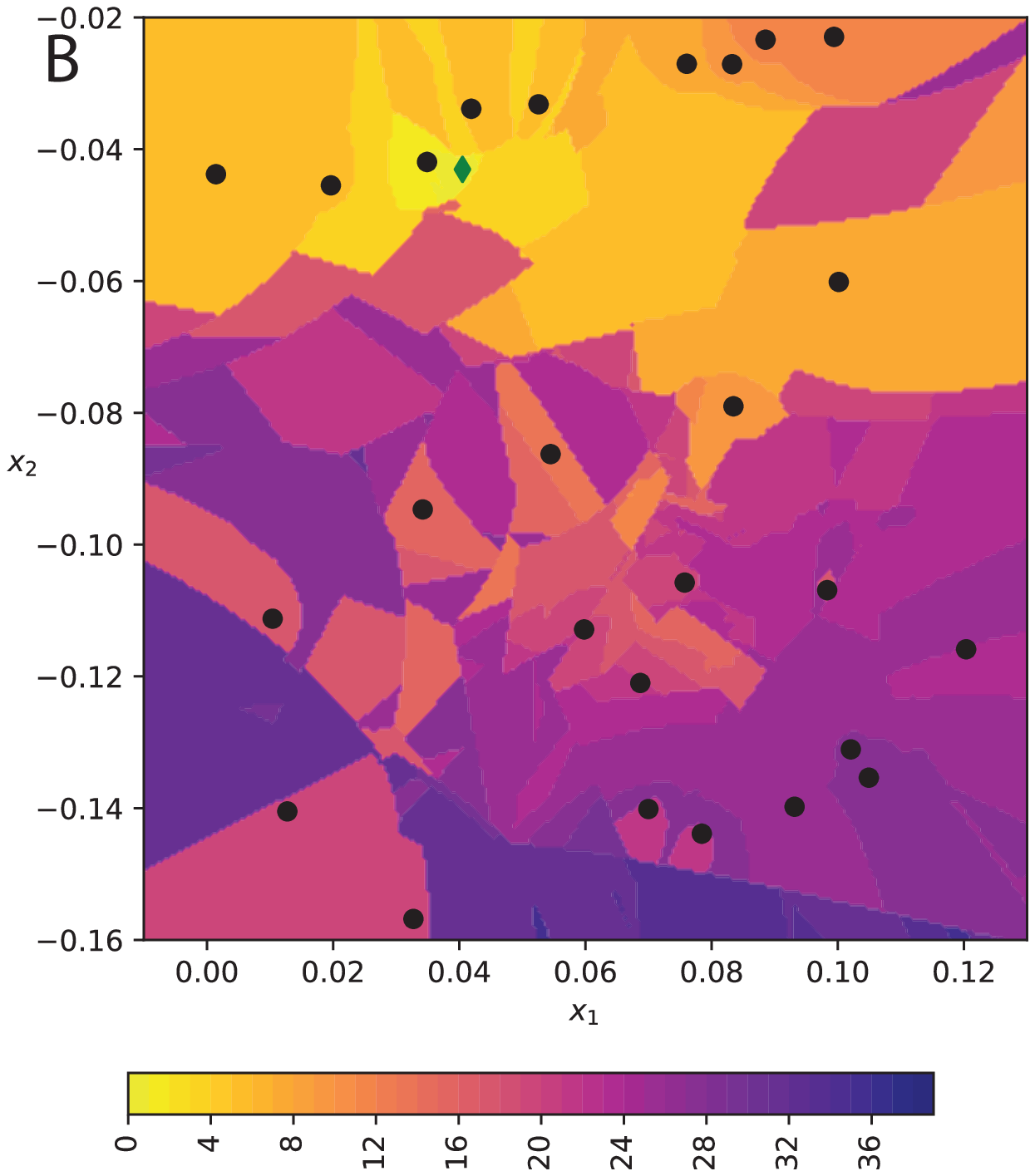}%
\includegraphics[width=.33\linewidth]{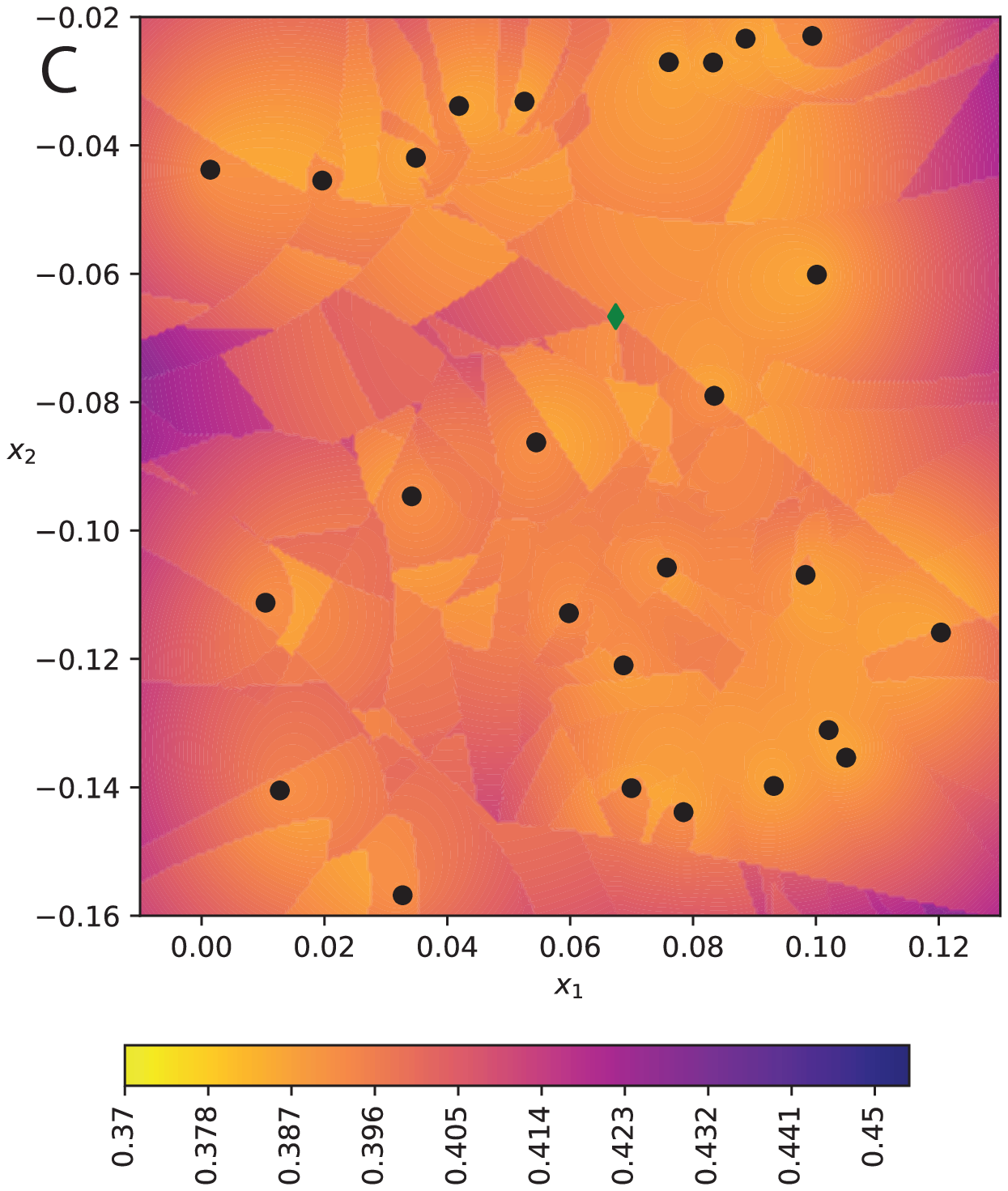}
\includegraphics[width=.33\linewidth]{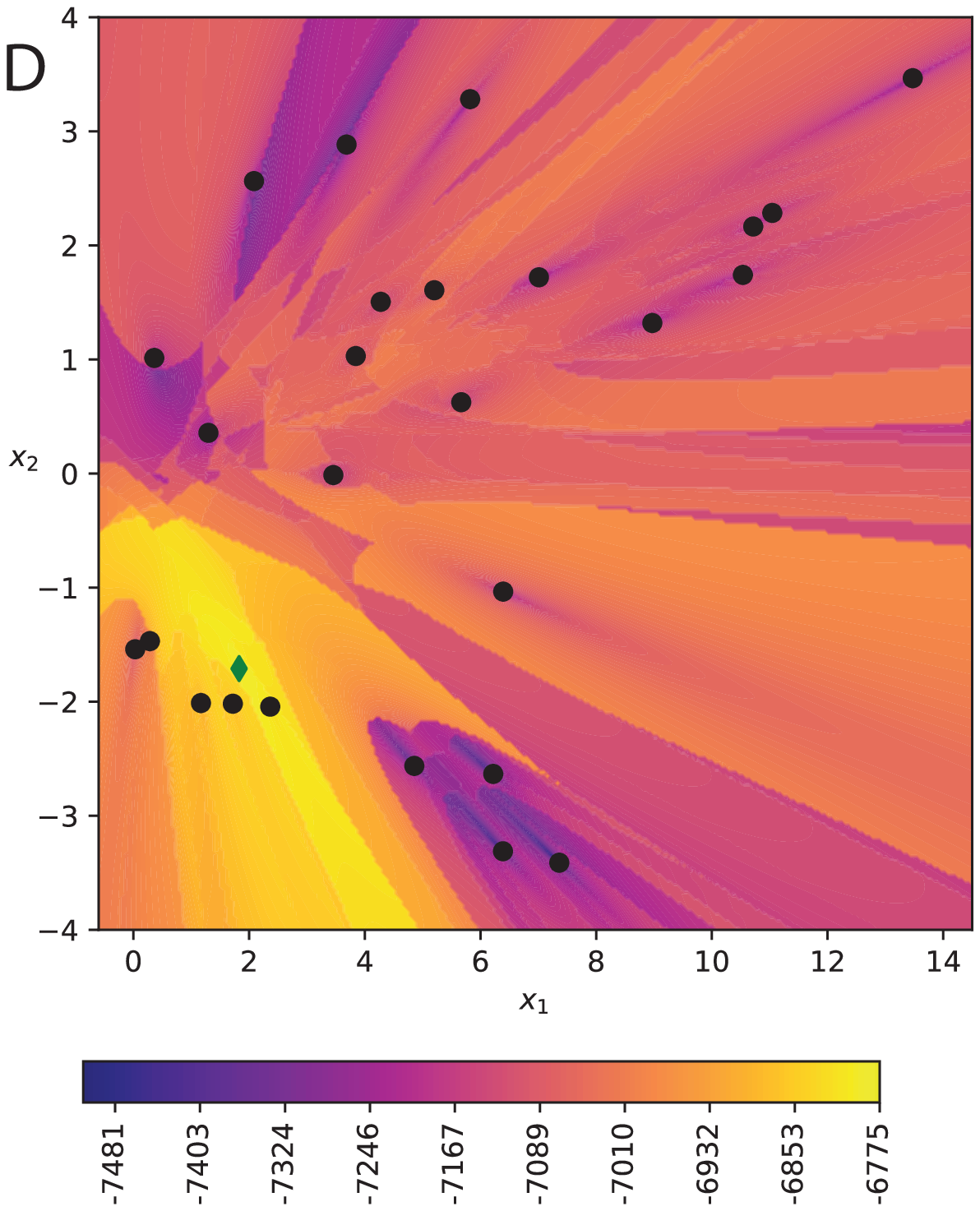}%
\includegraphics[width=.33\linewidth]{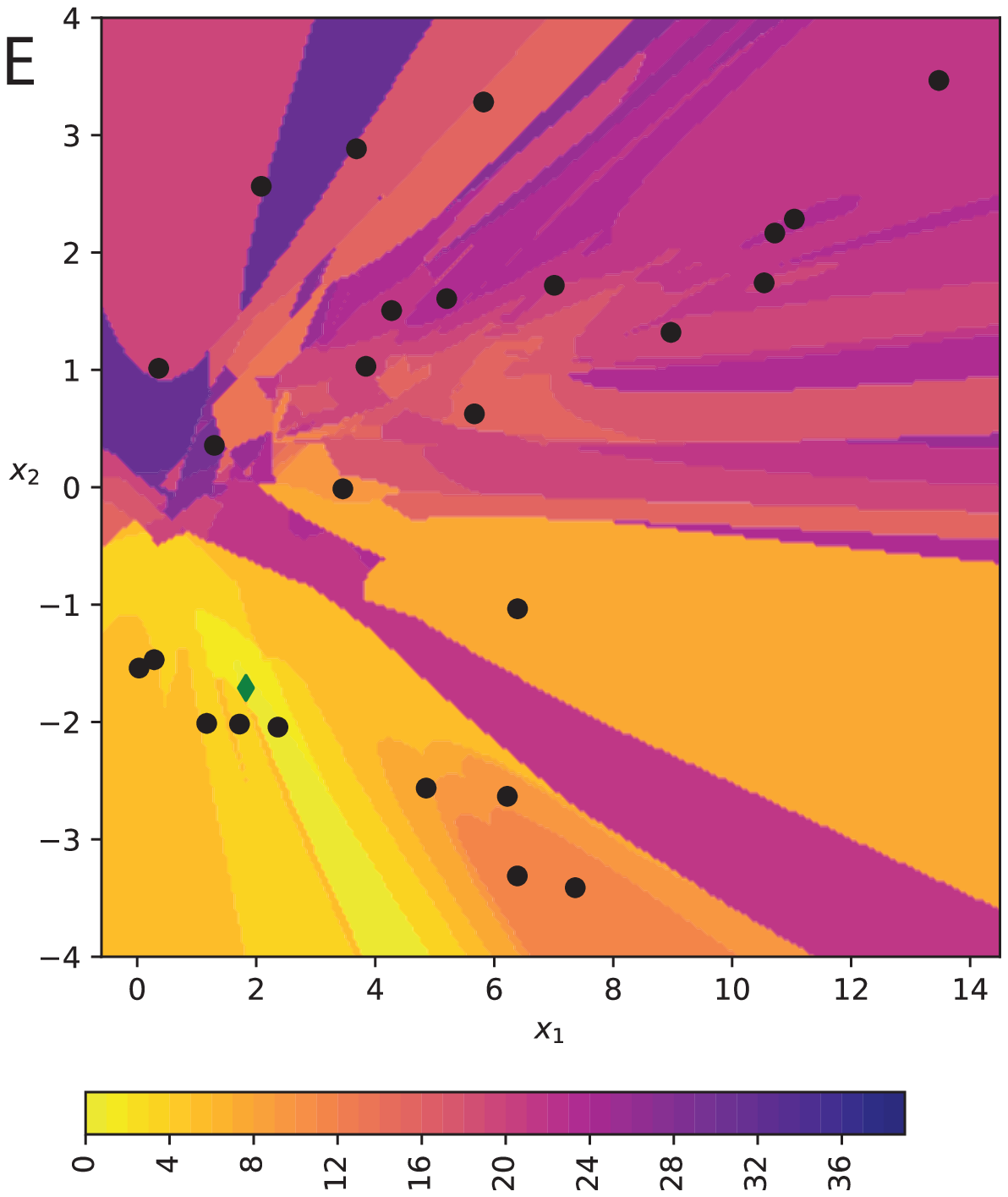}%
\includegraphics[width=.33\linewidth]{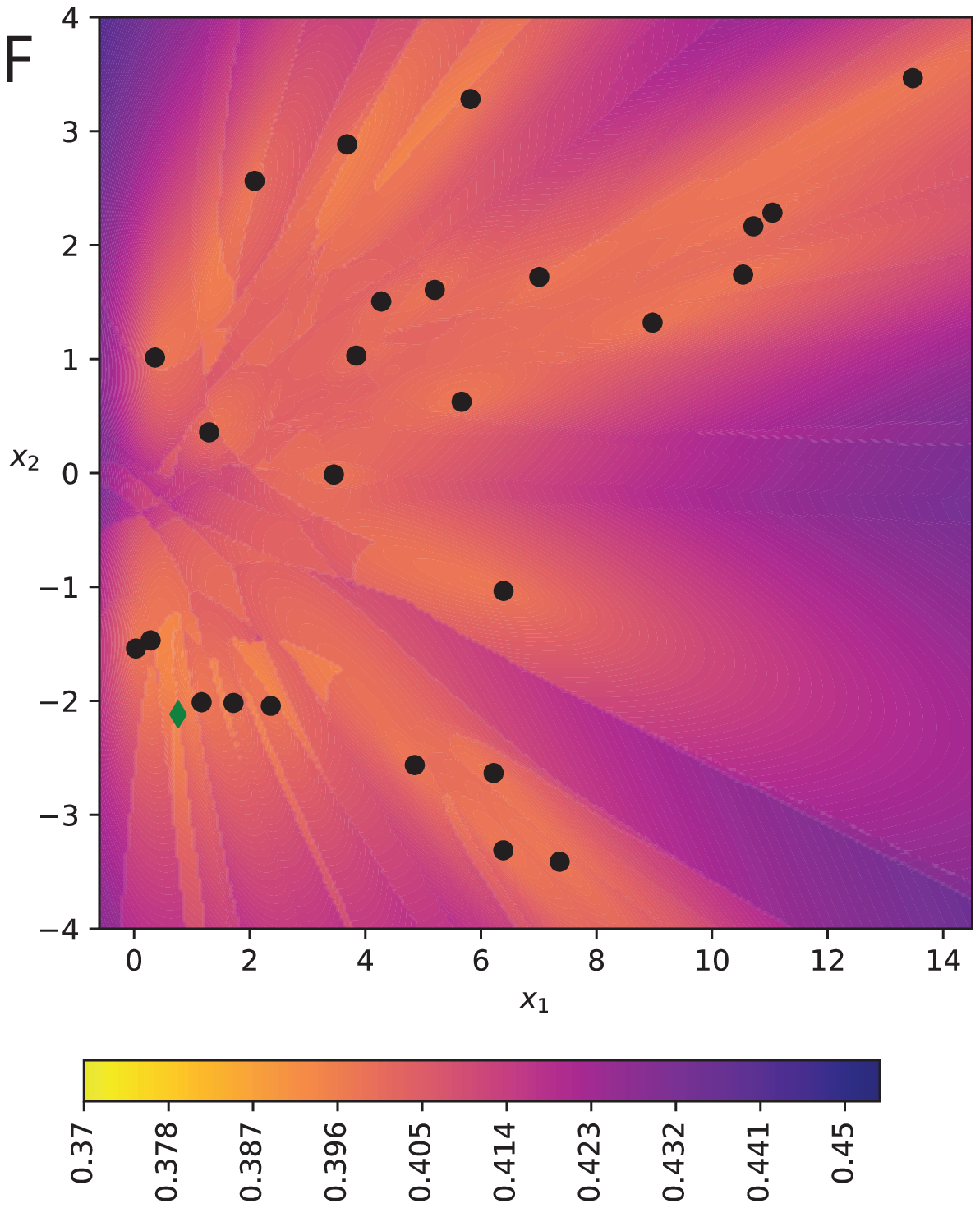}
\caption{Heat map of the of decoded tree's a, d) joint probability, b, e) symmetric difference from best tree topology and c, f) total length. Obtained by a grid search in $\mathbb{H}^{2}$, moving one node (green diamond) through embedding space.
Top row curvature $\kappa = -1$ and bottom row $\kappa = -1000$.
Black dots are the fixed locations of the remaining nodes.
}
\label{fig:grid}
\end{center}
\end{figure}

A heat map of the posterior in $\mathbb{H}^{d}$ ($d=2$) reveals that the posterior distribution is multifaceted and depends on the locations of the other nodes, figure~\ref{fig:grid}a).
It is possible to move between different topologies, however the posterior appears multi-modal and is only smooth piecewise, with distinct regions corresponding to different tree topologies.
This correspondence is apparent when compared with the middle panels, which shows the symmetric difference between the decoded tree to the optimal tree.
Coarsely, locations leading to a larger symmetric difference are further in the embedding space from the optimal location.

Consistent with corollary~\ref{cor:length_continuous}, the tree length surface is smoother with lower curvature, figure~\ref{fig:grid}cf).
With more negative curvature $\kappa=-1000$, changes in the posterior surface are less sharp, figure~\ref{fig:grid}d).
There are fewer jumps in the posterior and the regions have more uniform boundaries between them, figure~\ref{fig:grid}e).
This reflects the more local nature of the very curved space; moving one node has a smaller and more local effect on the tree.
Whereas distal nodes are almost unaffected by perturbing this node.
Sampling from a Normal distribution about the node could match the posterior surface better with this low curvature.

\subsection{Embedding Dimension}
The dimension of the hyperbolic manifold is also not prescribed and may affect the quality of the MCMC.
Figure~\ref{fig:dim} highlights that just three dimensions is optimal to perform Bayesian inference.
Increasing beyond three dimensions does not significantly affect the splits (ASDSF).
Similarly, the variance of the tree length is generally well estimated in three or more dimensions.
Two dimensions appears insufficient on both fronts to capture the phylogenetic posterior.

\begin{figure}[htbp]
\begin{center}
\includegraphics[width=.33\linewidth]{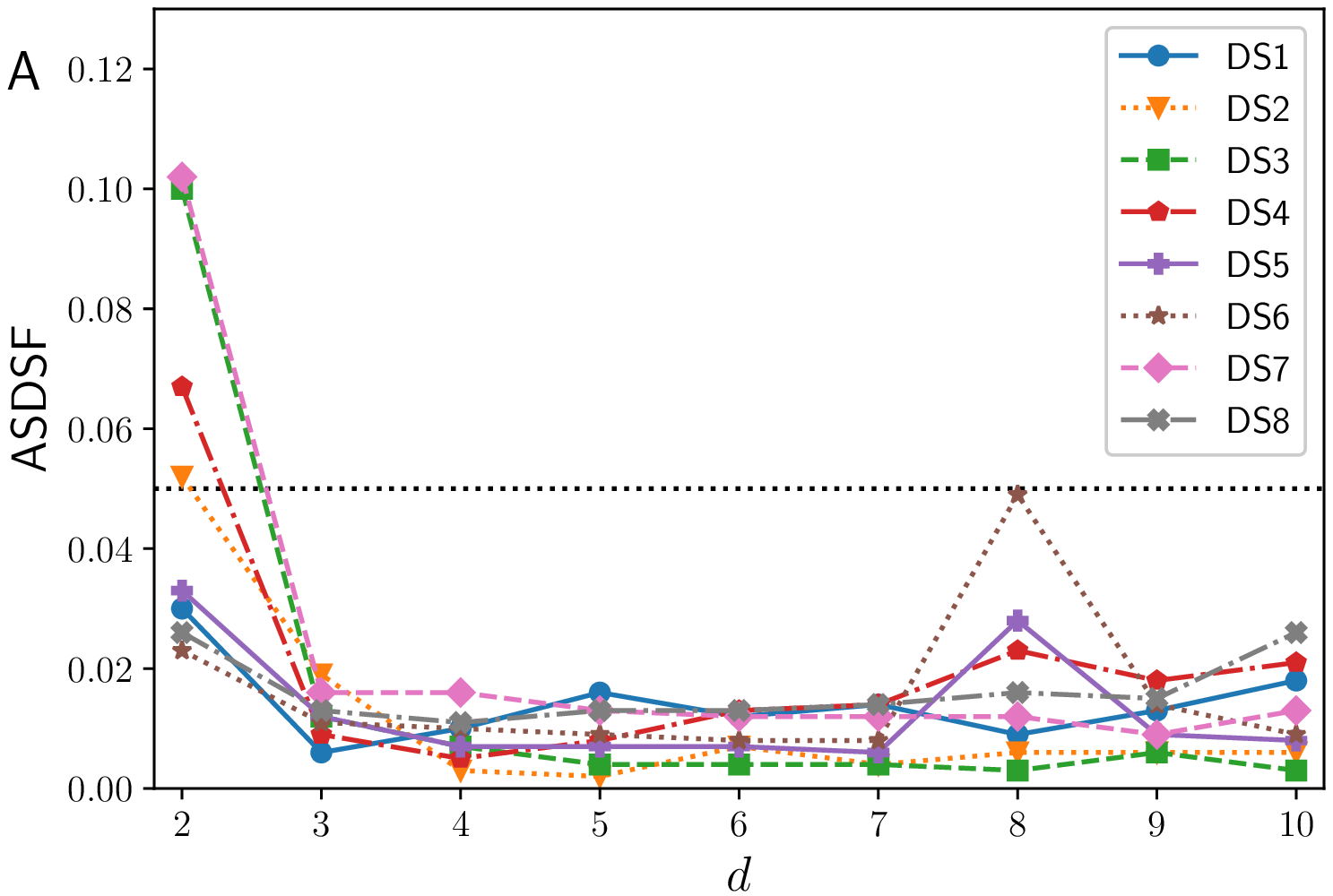}%
\includegraphics[width=.33\linewidth]{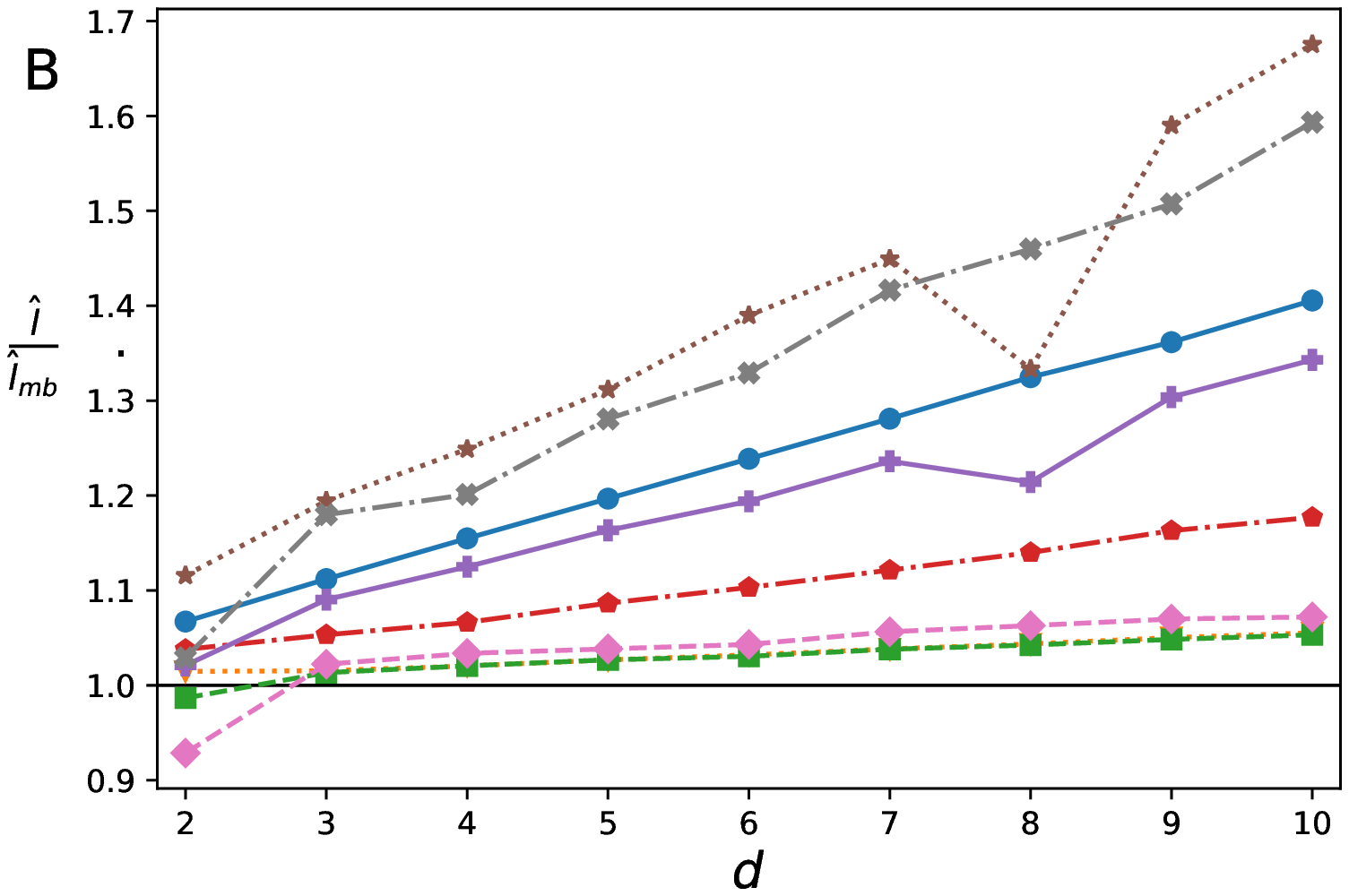}%
\includegraphics[width=.33\linewidth]{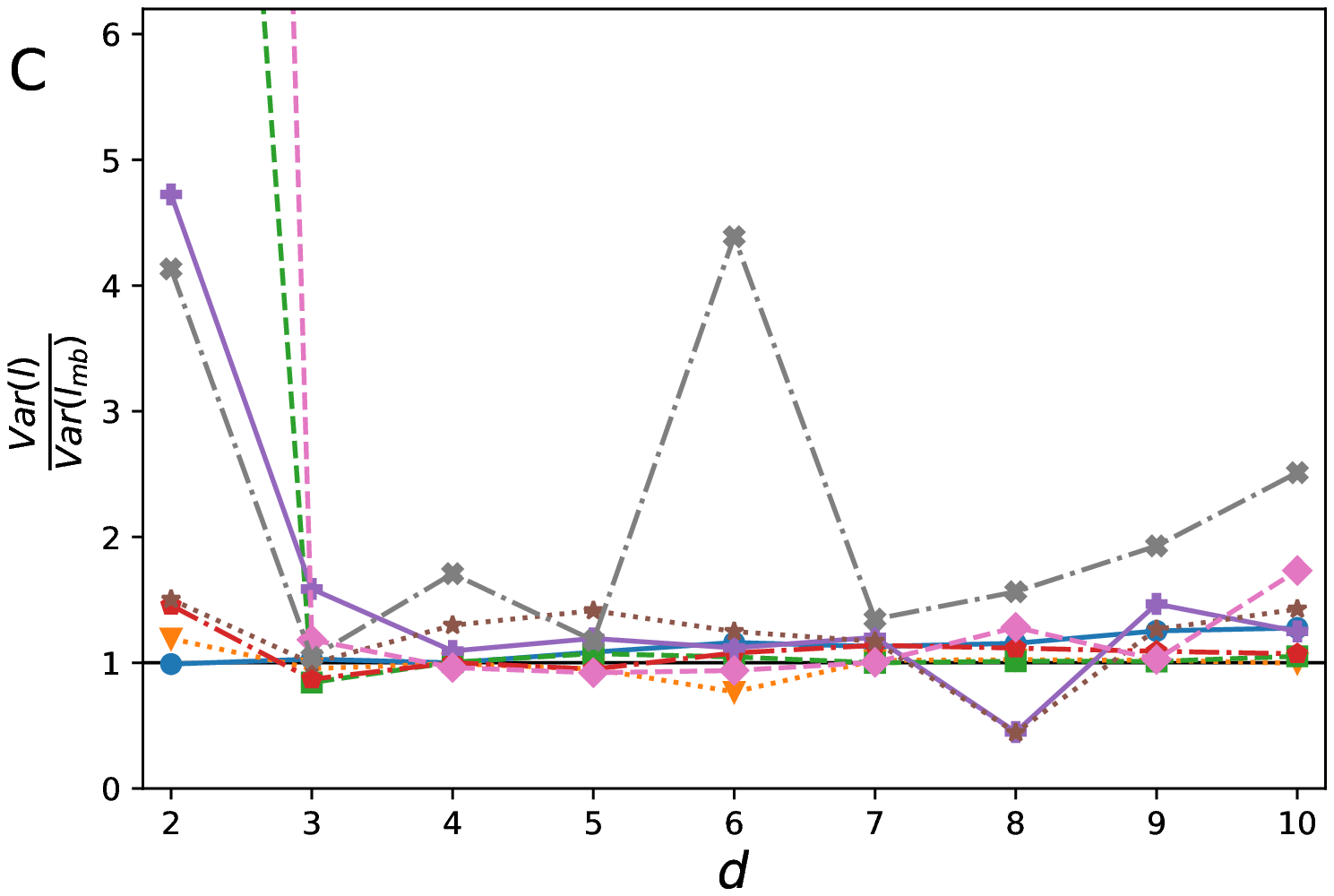}
\includegraphics[width=.33\linewidth]{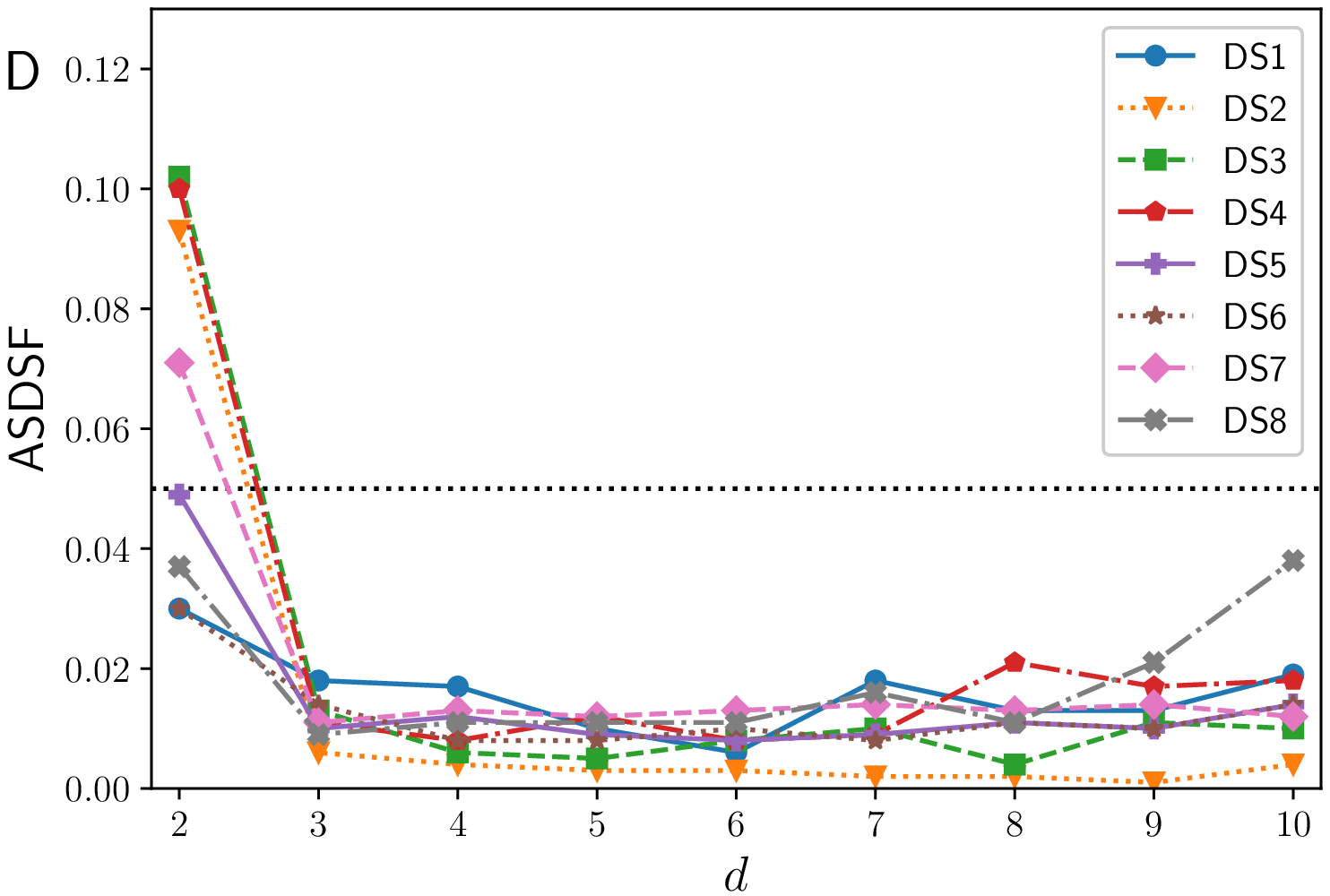}%
\includegraphics[width=.33\linewidth]{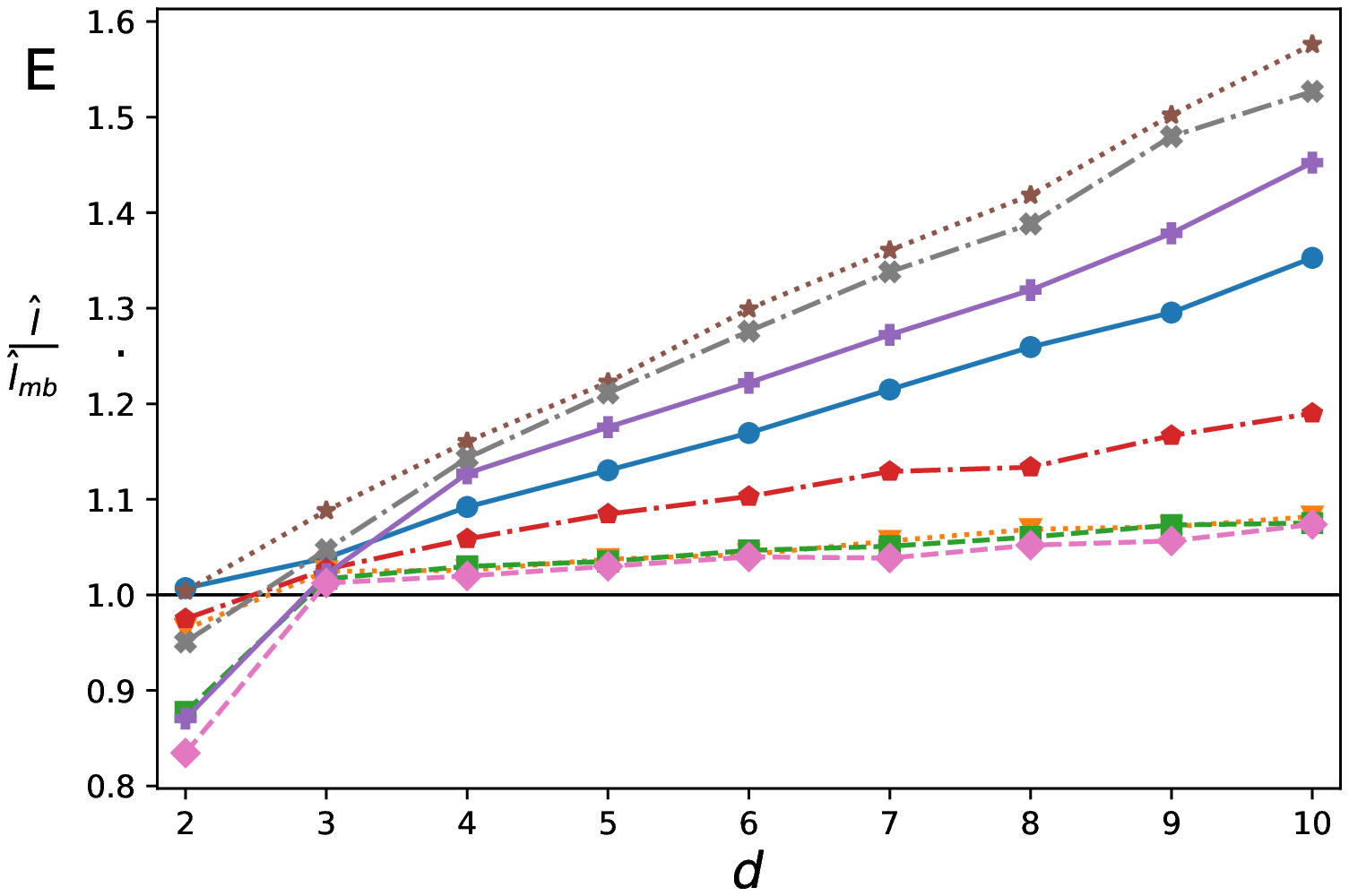}%
\includegraphics[width=.33\linewidth]{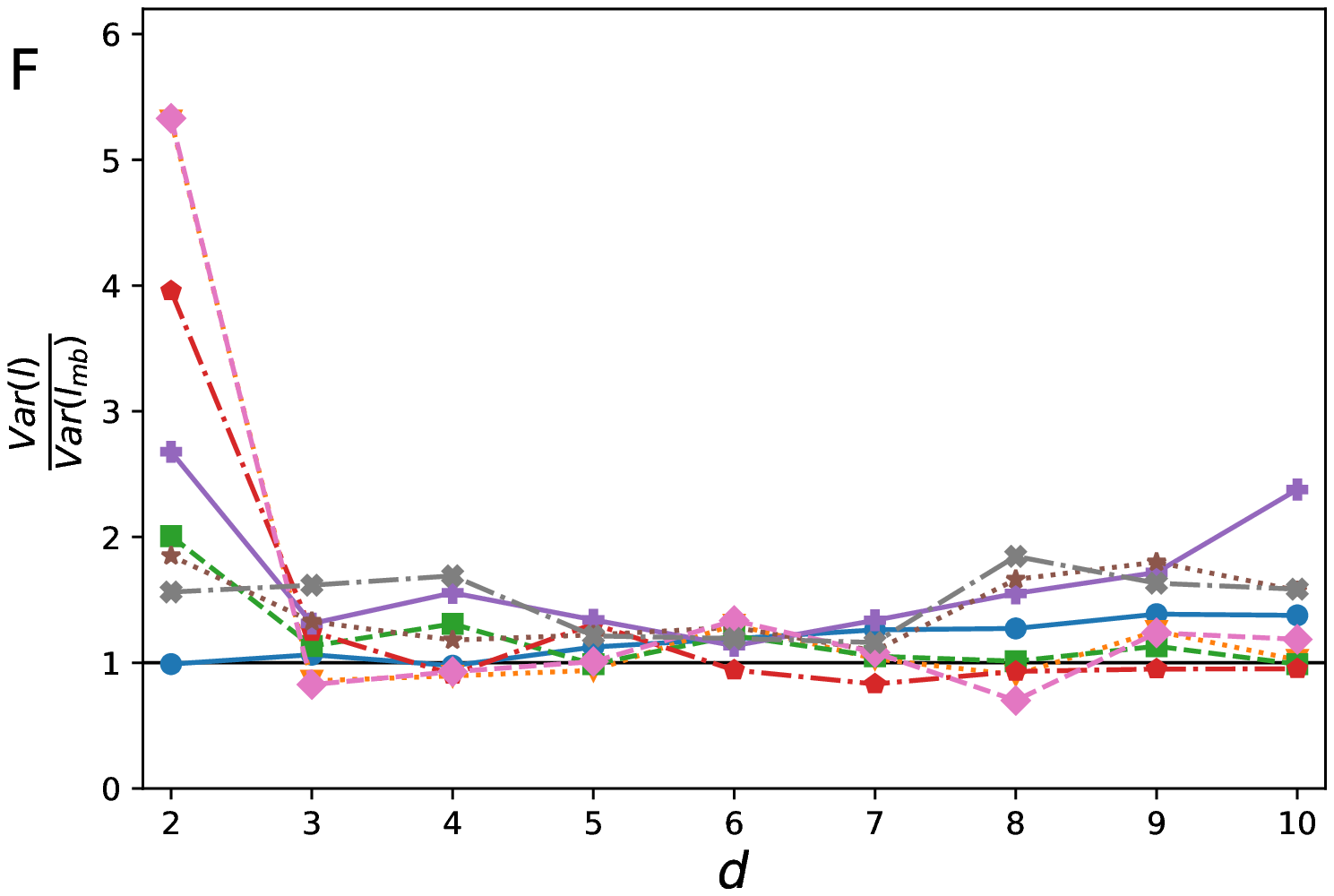}
\caption{Effect of embedding dimension on the posterior distribution.
Comparison to true posterior of:
a, d) ASDSF,
b, e) relative difference in median tree length,
c, f) relative difference in the variance of tree length.
Top row curvature $\kappa = -1$, bottom row $\kappa = -100$.
Truncated variance ratio in c) for DS3 with $d=2$ is $16.41$ and $39.05$ for DS7.
}
\label{fig:dim}
\end{center}
\end{figure}

Higher dimensions appeared to produce longer trees.
This effect appears linear in the embedding dimension and depends on the data set.
DS1, DS6 and DS8 are most affected, whereas DS2, DS3, DS4 and DS7 are least affected.
These groups correspond to the shortest and longest trees according to the median length from the golden runs respectively.
We checked that similar results are obtained with a lower curvature (Fig.~\ref{fig:dim} bottom row).

\subsection{Locality in Tree Space}
As Dodonaphy proposes new hyperbolic embeddings, the decoded tree topology freely changes according to $\text{NJ}: \{\bm{x}_{i}\} \mapsto T$.
To quantify this change we looked at the symmetric difference from the current tree state to the proposal tree for each generation.
This gives an idea of how topologically different proposal moves are from the current tree.
The symmetric difference measures how many splits are in one tree but not the other.
Looking at the cold chain in one MCMC run with $10^{6}$ proposals, most proposals make small or no changes to the tree topology, figure~\ref{fig:sym_diff}.
These include proposals from one heated chain, which are more likely to be in distant parts of the embedding space and could lead to larger symmetric differences.
Longer proposals occur less and less frequently, decaying roughly exponentially.

In addition to the topology one might expect that branch lengths within a topology could be updated far more efficiently with this approach than in classical phylogenetic MCMC where they are generally updated one-at-a-time.
This may matter more as trees grow large and practitioners start stacking in other continuous model parameters, as is done for viral phylodynamics.

\begin{figure}[htbp]
\begin{center}
\includegraphics[width=.6\linewidth]{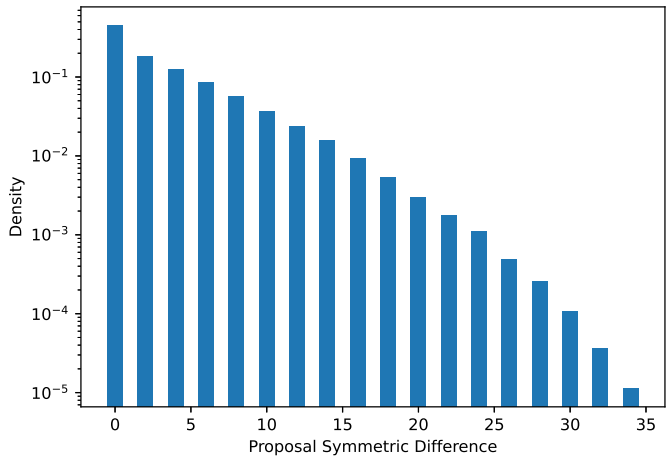}
\caption{Distribution of the proposal tree's symmetric difference to the current tree for DS1 in three dimensions with curvature $\kappa=-1$.
}
\label{fig:sym_diff}
\end{center}
\end{figure}

\subsection{Algorithm Run Time}
The three most time-consuming algorithms of Dodonaphy's MCMC are the likelihood calculation $O(nL)$, neighbour-joining $O(n^{3})$ and the pairwise distance computation $O(n^{2}d)$~\citep{chowdhary2018improved}.
We run Dodonaphy on Dell PowerEdge R640 or R740 servers with dual Intel(R) Xeon(R) Gold 6240 CPUs running at 2.60GHz.
For each run, we record the duration per chain at each generation $t$ over a range of taxa sets $n$ with varying unique site patterns $L$ and various embedding dimensions, figure~\ref{fig:times}.
Adding the three main algorithms together, the best fit with standard deviations is $a n^{3} +  b Ln + c n^{2} d$ with $a=5.086\pm 0.766 \times10^{-05} \,\text{ms}$, $b=2.529\pm 0.205 \times 10^{-04}\,\text{ms}$ and $c=6.884\pm 0.661 \times 10^{-04}\,\text{ms}$.
Assuming three dimensions and taxa with $L=1000$ unique sites, this means that the majority of the time is spent on neighbour joining when there are more than 237 taxa.

\begin{figure}[htbp]
\begin{center}
\includegraphics[width=.6\linewidth]{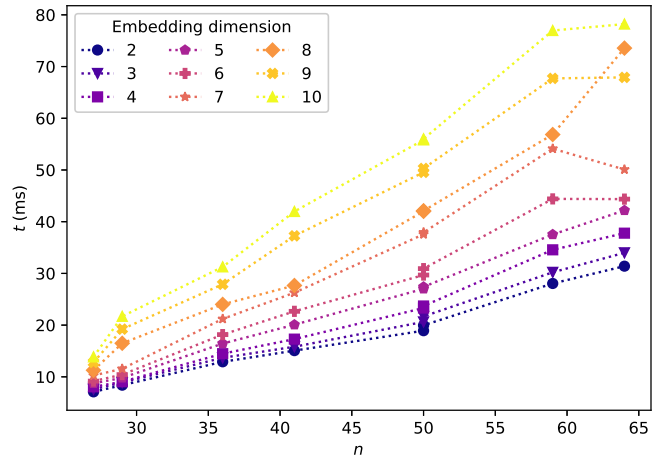}
\caption{Mean time per generation of one chain of Dodonaphy for each embedding dimension. Note that two datasets contain $n=50$ taxa.}
\label{fig:times}
\end{center}
\end{figure}

Overall, the expected performance complexity is between quadratic and cubic in the (finite) number of taxa.
Improvements to the NJ algorithm would yield the most significant speed-up as the number of taxa grows.
For example, there are fast NJ implementations suitable for hundreds of thousands of sequences~\citep{wheeler2009largescale}.
Similarly, BEAGLE provides a fast implementation of the likelihood computation~\citep{ayres2012beagle, ayres2019beagle}.
Nevertheless, asymptotically, NJ remains the algorithmic bottleneck of Dodonaphy with $O(n^{3})$.

\bigskip

\section{Discussion}
\label{sec:discuss}
\subsection{Fidelity Of Embedded MCMC}
The results indicate that Hyperbolic space offers fruitful embeddings for Bayesian phylogenetics.
A standard MCMC algorithm can use hyperbolic space to navigate the posterior phylogenetic landscape and produce the posterior tree distribution.
It can start from an embedding initialised by the NJ tree on the sequence alignment to arrive at the posterior.
Once converged, it recovers similar splits frequencies, split lengths and tree lengths to the \textit{ground-truth}.

As expected, hyperbolic embeddings performed significantly better than Euclidean embeddings.
In Euclidean space, phylogenies can be embedded by taking the square root of their distances~\citep{devienne2011euclidean}.
However, embedding $n$ taxa could require as many as $n-1$ dimensions~\citep{layer2017phylogenetic}.
This makes it less scalable and less suitable for the large datasets being produced by infectious disease surveillance programs.

Previous work on maximum likelihood phylogenetics via hyperbolic embeddings has shown that the fidelity of the embedding improves as the embedding curvature decreases~\citep{wilson2021learning}.
In line with this existing work but in a Bayesian context, we find that decreasing the curvature from zero improves the MCMC quality.
This improvement saturates at about $\kappa = -1$, indicating that the error to the four-point condition $\delta$ is negligible compared to the branch lengths.
Decreasing beyond this to a wide range of curvature $[-100, -1]$ provided good results, indicating the suitability of hyperbolic space for phylogenetic embeddings.

Lower curvatures tended to produce smaller tree length variances than the golden run.
Underestimated variances could signify that the MCMC chains are contained in local optima, which is consistent with the higher split deviation (ASDSF) from the golden run.
This suggests that there is an optimal range of curvature for embedding phylogenies, balancing between competing factors.
Decreasing the curvature makes the tree length and, to some degree, the posterior smoother.
However, a very curved embedding space becomes too localised and it is difficult to make significant changes to the tree topology.

Overestimated tree length variances indicate the MCMC is proposing trees outside the credible set, and sampled trees are either too long or too short compared to the credible set.
Generally this occurred for embeddings with an almost Euclidean curvature (close to zero), where the error on the four-point condition is larger.
This could indicate that the proposal distribution is not matching the posterior landscape well.
A ``neighbouring'' topology in the embedding space with low posterior probability could be over-sampled in these finite MCMC runs.
However, with a more curved space ($\kappa=-100$) this did not occur as significantly for three or more dimensions.

Just three hyperbolic dimensions appears sufficient to continuously represent phylogenetic trees.
Higher dimensions led to overestimated tree lengths, especially for short trees.
This overestimation occurred regardless of the embedding curvature and appeared linear in the embedding dimension.
This result highlights the suitability of low dimensional hyperbolic embeddings, regardless of the number of taxa.

One possible explanation of the tendency to overestimate branch lengths is the Jacobian involved in transforming proposal samples in hyperbolic space to samples in tree space, via neighbour joining.
%MF: what if we calculate the Jacobian without the final transformation NJ, with a bit of luck it might help
The Jacobian comes from each of a series of three transformations, starting from the sample space $\mathbb{R}^d$ and ending in tree space.
Points are sampled in $\mathbb{R}^{nd}$ and 1) projected onto the hyperboloid $\mathbb{H}^{nd} \subset \mathbb{R}^{n(d+1)}$ via $\phi$, 2) then these points are transformed into distances in $\mathbb{R}^{n \choose 2}_{\geq 0}$ before 3) finally being transformed into a tree by neighbour joining.
Closed forms for the Jacobians of the first two are given in appendix~D.

However, the determinant of the Jacobian of the second transformation is necessarily zero because for each of the ${n \choose 2}$ the distances, only $2(d+1)$ of the inputs (the two embedding points) contribute to the distance.
The derivative of each distance is zero with respect to the other $(n-2)(d+1)$ embedding locations.
As we discuss below, multiple embedding configurations may produce the same set of distances, yielding a non-invertible transformation, consistent with a null Jacobian determinant.
Further, computing the Jacobian of neighbour joining is non-trivial because it requires the $\min$ function.
These Jacobian adjustments are omitted from Dodonaphy's algorithm.

%MF: if it still does not fix the long branch problem, a prior on the distances could be another option
One remedy for this issue is to use a prior on the embedding locations rather than on trees.
We demonstrate that using a standard multivariate Normal as a prior does alleviate this trend in appendix~C.
A drawback of the multivariate Normal prior is that it is an uncommon prior for phylogenetics, making it difficult to compare to state-of-art methods.
%AD 220604: nice one! any idea if it would be possible to stack another distribution on top of the multivariate normal as a hyperprior? Or does that break the math and bring us back to the Jacobian? Trying to think up some way of making a hyperprior that imposes a distribution on tree length, although it's not obvious quite how that could be achieved. Maybe something on the distribution of geodesic distances among embedding points? wouldn't be nice to compute though - O(n^2)
%MM 220606: There might be a way to do it, but the point is just to demonstrate that when we for sure don't need the Jacobian, the linear trend in the tree length disappears.

\subsection{Tree Considerations}
Embedding the internal nodes of a tree could lead to more degrees of freedom to explore tree space faster.
This is simple for a fixed, given topology.
However, connecting the internal nodes into a consistent bifurcating tree is a challenging problem.

Multi-furcating trees could be decoded using a (minimum) spanning tree or the approach used by SLANTIS~\citep{koptagel2022vaiphy}, which connects internal nodes using Bernoulli trials.
The efforts of~\citet{friedman2001structural} yielded a way to convert a multi-furcating tree into a bifurcating tree with the same likelihood.
One could make use of this algorithm after forming a minimum spanning tree of the embedded points.
However, the method does not maintain consistency: distances in hyperbolic space between internal nodes wouldn't translate to distances on the tree.

Dodonaphy constructs unrooted NJ trees, which is useful for exploring the relationship between taxa.
If a rooted time tree is desired, other clustering algorithms such as the UPGMA or WPGMA would produce rooted trees from the taxa embeddings.
However, these algorithms require ultra-metric distances to output a unique rooted tree, a property that hyperbolic space doesn't guarantee.
If the distances are ultra-metric, one could also add an MCMC parameter for the rate of evolution to map from evolutionary distance to time.
In this case, distance on the hyperboloid would represent time rather than substitutions per site.

\subsection{Future Research}
%MF: just some comments:
%MF: I wonder if the curvature could be tuned during mcmc, it is just a scalar. I imagine getting good mixing while tuning the number of dimensions would be much harder since that would involve transdimensional proposals
%MM: I'm not sure how to directly tune the curvature because I'm not sure what feedback you'd use to increase or decrease the curvature. However, we could allow a new curvature to be a part of each proposal.
%MM: For transdimensional proposals I suppose it wouldn't be too hard to delete a dimension as a proposal. You could use a prior on the embedding locations to sample locations in any extra dimensions. But it might be hard to get a good tree that's accepted.
%MF: Instead of using the regular MCMCMC using different temperatures we could have tried multiple chains with different dimensions or curvatures
%MM: Good idea, it relatively simple to have a different dimension and curvature for each chain.
\subsubsection{Beyond JC69}
The JC69 model has no free parameters, so on a given topology, the branch lengths alone determine the model likelihood.
Substitution models with free parameters, such as the generalised time reversible substitution model~\citep{tavare1986probabilistic}, are not directly encoded by the embedding.
Incorporating these additional model parameters could be done as for a regular MCMC, rather than in an embedding space.

\subsubsection{Embedding Isometries}
Isometries of hyperbolic space can move an embedding to numerous locations and produce the same tree.
These form an equivalence class on trees $[T]$.
The isometries of hyperbolic space are given by the orthochronous Lorentz group $O^{+}(1, d)$.
These are characterised by the $d+1 \times d+1$ orthogonal matrices $g$ such that
\be
g^{\mathsf T}Hg = H.
\ee
For example, a hyperbolic rotation in $\mathbb{H}^{2}$ by angle $\alpha$ is 
\be
g = 
\begin{bmatrix}
\text{cosh}(\alpha)& \text{sinh}(\alpha) \\
\text{sinh}(\alpha) & \text{cosh}(\alpha)
\end{bmatrix}.
\ee
Rotating all the tip locations would leave the tip distances unaltered and decoded to the same tree.

In addition to the symmetries of hyperbolic space, the tree structure may possess symmetries.
For example, swapping the position of two taxa in a cherry could leave the decoded tree unaltered.
Accounting for symmetries in hyperbolic space and a decoded tree determines the equivalence classes of tree embeddings and is, to our knowledge, an open question.
%AD 220604: thought provoking section! or subsubsection I guess. I wonder how it might relate to the ability of an MCMC to mix or gradient descent to traverse the space. Could some isometric configurations produce less "jagged" likelihood surfaces than others? Not suggesting any manuscript edits, just leaving a comment.

\subsubsection{Local Tree Moves}
When the curvature is sufficiently small, neighbour joining decodes trees where the taxon-taxon distances in the decoded tree are arbitrarily close to the embedding distances, corollary~\ref{cor:length_continuous}.
As a result, small perturbations of the locations of embedded sequences produce small changes to the taxon-taxon distances in the decoded tree.
This provides a natural way to move ``locally'' in tree-space: comparing the taxon-taxon distances, rather than the discrete number of NNI moves.

Dodonaphy proposes new MCMC states by continuous probability distributions.
Here it uses a multivariate Normal distribution for proposals, however other distributions could equally be used.
We discuss an alternative proposal with similar results in appendix~B.
Matching the proposal distribution to the posterior is sought for the efficiency of MCMC chains.
Exploring additional proposal distributions could lead to a better match.

This representation of the tree space opens up the possibility of exploring the neighbourhood of trees.
Dodonaphy's representation of tree space bypasses classical topology changes such as NNI and SPR.
Relating new moves found by Dodonaphy to existing moves (NNI, SPR) could improve methods that need to propose trees, such as MCMC.

\subsubsection{Adding Taxa To an Existing Phylogeny}
Online viral phylogenetics takes genetic sequences as they arrive, i.e. sequentially.
Including these new sequences into trees in the posterior space is a matter of active research~\citep{matsen2010pplacer, fourment2018effective}.
Dodonaphy could offer the ability to place a new taxon onto a tree in constant time via the embedding space.

One simple way to approximate the embedding location $z^{'}$ of a new taxon labelled $k$ is to use its distances to a subset of the taxa.
Select a subset $Z \subseteq X$ of taxa with cardinality one more than the embedding dimension $|Z| = d + 1$.
Compose the matrix $S_{ij} = z_{i}$ containing the taxa locations on the hyperboloid $z_{i}\in \mathbb{H}^{d} \subset \mathbb{R}^{d+1}$.
Then, the genetic distances from the new taxa $k$ to the selected subset are $D_{i, k} \in \mathbb{R}^{d+1}_{\geq 0}$ for $i= 0, 1,...,  d+1$.
These distances satisfy
\be
S z^{'} = \text{cosh}(D)
\ee
which has solution $z^{'} = S^{-1}\text{cosh}(D)$ in $\mathbb{R}^{d+1}$.
This solution could be mapped onto the hyperboloid sheet by selecting the closest point, or simply adjusting the first coordinate using eq.~\ref{eq:x0}.
Incorporating new taxa in this way could lead to a fast method of adding taxa to a phylogeny.

For more accuracy, $S$ could be extended with additional taxa to become an over-determined system and solved approximately by its least squares solution $z^{'} = \left(S^{\mathsf {T}}S\right)^{-1}S^{\mathsf {T}}  \text{cosh}(D)$.
Alternatively, the subset of taxa $Z$ could be reselected multiple times, effectively working as mini-batches.
Then the most likely placement could be used based on the reconstructed super tree's probability.

\subsubsection{Differentiable Embeddings}
An exciting avenue for future research is the potential for gradient-based phylogenetic methods in the embedding space.
Differentiating the posterior probability as a function of the embedding locations would allow this.
However, gradients cannot flow through tree decoding by NJ as it is not differentiable.
The difficulty lies in selecting which two taxa to connect during NJ, which requires using the (non-differentiable) minimum function.

Decoding a differentiable tree would provide a gradient-based framework for exploring tree space.
This would be useful for point estimates of the maximum \textit{a posteriori} tree or the maximum likelihood tree.
It could also allow gradient-based Bayesian methods such as Hamiltonian Monte Carlo (HMC) and variational inference, which we discuss below.

Hamiltonian Monte Carlo assigns a mass to each tip position in the embedding space and draws on Hamiltonian mechanics to cover more of the posterior space in fewer iterations.
For tree space, moving through the discrete barrier between topologies with momentum is still challenging.
\citet{dinh2017probabilistic} propose a high dimensional $d=(2n-3)!!$ solution using a surrogate function to smoothly bridge the potential energy between topologies.
Dodonaphy offers a low dimensional embedding space for HMC.
Whilst the posterior may only be piecewise smooth in the embedding space, some methods cope with these transitions~\citep{mohaselafshar2015reflection, dinh2017probabilistic}.

Variational inference approximates the posterior distribution with much simpler analytical distributions.
The goal is to minimise the error of the approximation via optimisation.
Trees are discrete objects so this is usually considered a categorical distribution.
For example, recent works appeal to graphical models to encode such categorical distributions~\citep{zhang2018generalizing, zhang2019variational}.
In contrast, continuous distributions for the location of each node in the embedding space provide continuous distributions in tree space.
Making use of distributions in hyperbolic space could open up an alternative technique for gradient-based variational inference.

\section{Conclusions}
\label{sec:conc}
Hyperbolic space offers a fruitful embedding space for Bayesian phylogenetics.
Performing MCMC on eight data sets captured the splits (ASDSF $< 0.05$) and median tree lengths within $20\%$ on all datasets.
Embedding in three dimensions with a curvature between $-100 \leq \kappa \leq -1$ is sufficient to attain these results.
The MCMC can start from an embedding of a NJ tree on the sequence data and move throughout hyperbolic space to explore the posterior space of trees.

Phylogenetic embeddings allow MCMC algorithms to propose new states (topology and branch lengths) from continuous probability distributions, which provides a novel way to make ``local'' changes in tree space.

\section{Funding}
This work was supported by the Australian Government through the Australian Research Council (project number LP180100593).

%If you have any acknowledgements, please include them here.
\section{Acknowledgements}
Computational facilities were provided by the UTS eResearch High Performance Computer Cluster.

\section{Conflict of Interest}
All authors declare that they have no conflicts of interest.

%If your paper has accompanying supplementary data, please include the below statement in your PDF.
%\section{Supplementary Material}
%Data available from the Dryad Digital Repository:
% \href{http://dx.doi.org/10.5061/dryad.[NNNN]}%
% {http://dx.doi.org/10.5061/dryad.[NNNN]}.
%\url{http://dx.doi.org/10.5061/dryad.[NNNN]}.

\bigskip\bigskip

%%%%%%%%%%%%%%%%%%%% REFERENCES %%%%%%%%%%%%%%%%%%

% The best way to enter references is to use BibTeX.
\bibliography{ms}

%%%%%%%%%%%%%%%%%%%%%%%%%%%%%%%%%%%%%%%%%%%%%%%%%%

% Please include any figure captions on a separate page after the references. Figures themselves should be embedded in the text.

%\begin{figure}[!p]
%\centering\includegraphics{fig1}
%\caption{Figure caption}
%\label{Fig1}
%\end{figure}

%\begin{table}[!p]
% 1. Table titles should be in caps and lowercase
% 2. Footnotes can be used in Tables (a,b,c)}
%\tblcaption{Table title
%\label{Table1}}
%{\tabcolsep=4.25pt
%\begin{tabular}{@{}cccccccccc@{}}
%\tblhead{Heading & Heading & Heading & Heading & Heading}
%Value & Value & Value & Value & Value 
%\lastline
%\end{tabular}}
%\end{table}

%If you have any print appendices, please include them at the end of the document.
\appendix

\section{A: Tuning MCMC} \label{app:tune}
A simple method for tuning the covariance matrix in MCMC to achieve a target acceptance rate $a^{*}$ can be achieved as follows.
Given an initial covariance matrix $\Sigma$, we scale it by a factor $s$ that is tuned.
Assume that increasing the step size decreases the acceptance rate and vice versa.
We solve the following ordinary differential equation for the step size $s$:
\be
\frac{ds}{da} = a - a^{*}
\ee
using the Euler method:
\begin{equation} \label{eq:euler}
s_{n+1} = s_{n} + \eta (a_{n} - a^{*}).
\end{equation}
The learning rate $\eta = (1+n)^{-\lambda}$ decays to zero.
This ensures that as $n \to \infty$ the step size is constant and the target distribution is sampled.
We make the simple choice of the small multiple $\zeta=0.1$ of identity matrix as the initial covariance and choose $\lambda = 0.5$.

This method is similar to the adaptive scaling metropolis (ASM) algorithm~\citep{atchade2005adaptive} with a small difference.
In Eq.~\ref{eq:euler}, replacing the acceptance rate over all past states $a_{n}$ by the acceptance probability of the current proposal $\alpha$ reproduces ASM.
We find that using $a_{n}$ leads to better performance for the warm up phase, giving faster convergence and an acceptance rate closer to the target acceptance of $0.234$.

\section{B: Wrapping onto the Hyperboloid} \label{app:wrap}
\citet{nagano2019wrapped} present a method to wrap Euclidean vectors onto the hyperboloid.
This allows Dodonaphy to sample vectors drawn from arbitrary distributions in hyperbolic space.
We provide a brief outline of this method and use it for proposals within Dodonaphy to compare results.

The method has two steps, parallel transport of vectors tangent to the hyperboloid and then exponentially mapping from the tangent space onto the hyperboloid.
The tangent space of a point $\mu$ on the hyperboloid is the set of vectors on the tangent plane at $\mu$.
These are the points orthogonal to $\mu$ in $\mathbb{R}^{d+1}$:
\be
T_{\mu}\mathbb{H}^{d} := \{x \in \mathbb{R}^{d+1}: \langle x, \mu \rangle= 0\} .
\ee
Parallel transport takes a vector in one tangent plane $x \in T_{\nu}\mathbb{H}^{d}$ and places it in another tangent plane $T_{\mu}\mathbb{H}^{d}$ whilst conserving the vector's metric.
It is given by:
\be
\text{PT}_{\nu \to \mu} (x) = x + \frac{\langle \mu - \alpha \nu, x \rangle}{\alpha + 1} (\nu + \mu)
\ee
where $\alpha = -\langle \nu, \mu \rangle$

Then, the exponential map wraps this vector in the tangent space onto the hyperboloid.
It is constructed to preserve the vector's norm $||x||_{\mathcal{L}} = \sqrt{\langle x, x \rangle}$ as follows
\be
\exp_{\nu}(x) = \text{cosh}(||x||_{\mathcal{L}})\nu + \frac{x}{||x||_{\mathcal{L}}} \text{sinh}(||x||_{\mathcal{L}}).
\ee

As an example, Dodonaphy samples Euclidean vectors from a Normal proposal distribution and wraps them onto the hyperboloid.
This is instead of projecting the first coordinate of the hyperbolic vector.
Like before, MCMC is run for $10^{7}$ generations and all else is held equal with $\kappa=-1$, $d=3$.

\begin{figure}[htbp]
\begin{center}
\includegraphics[width=0.33\linewidth]{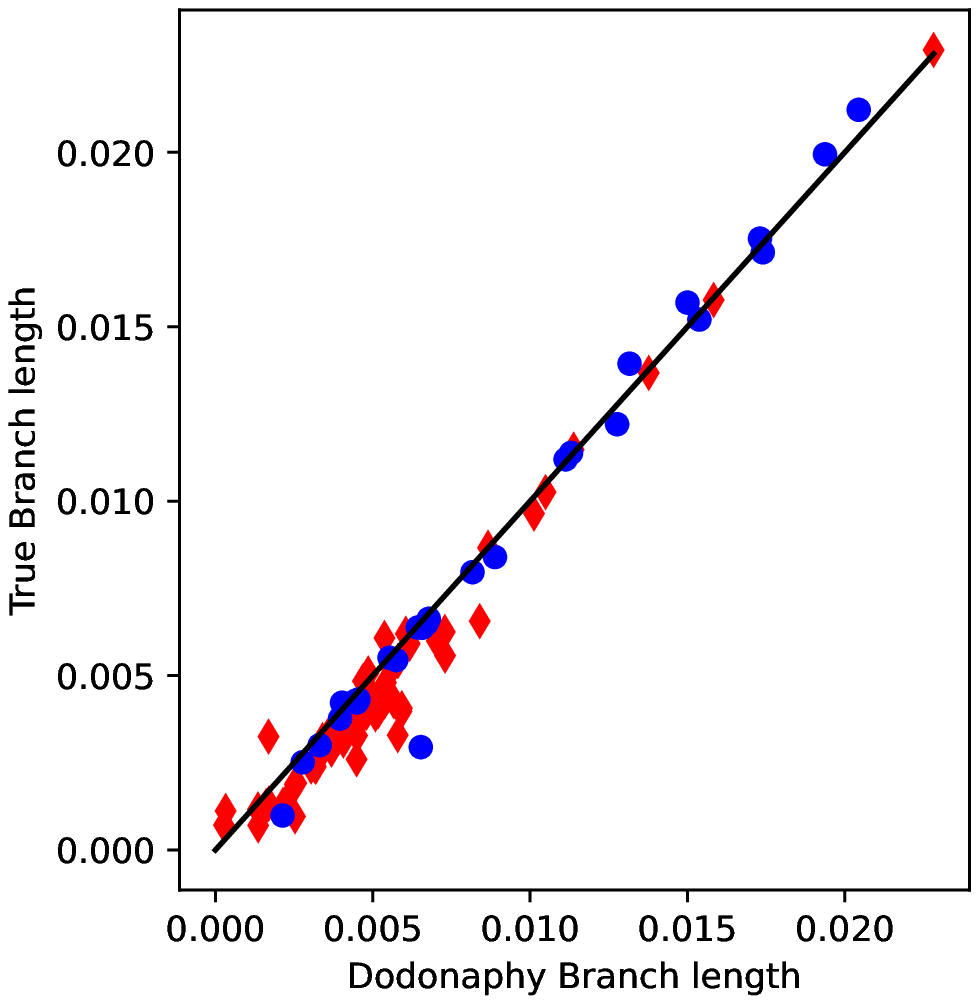}%
\includegraphics[width=0.33\linewidth]{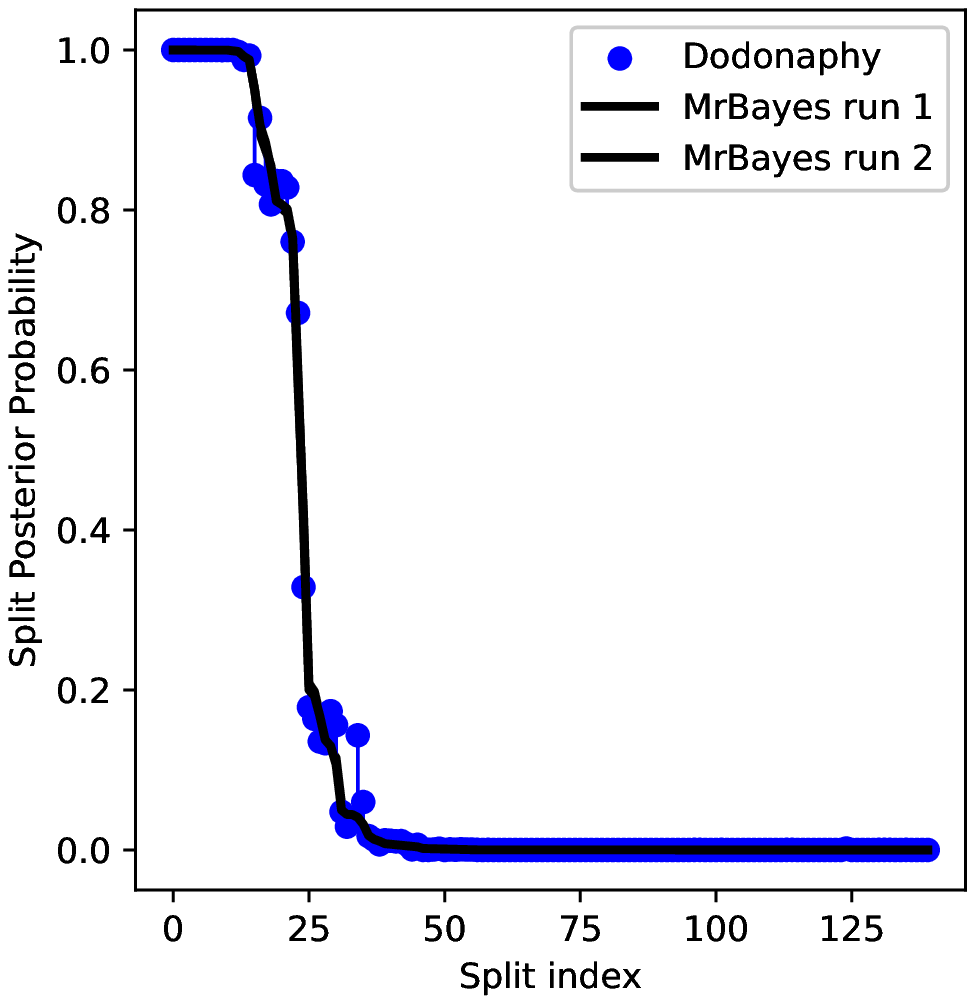}%
\includegraphics[width=0.33\linewidth]{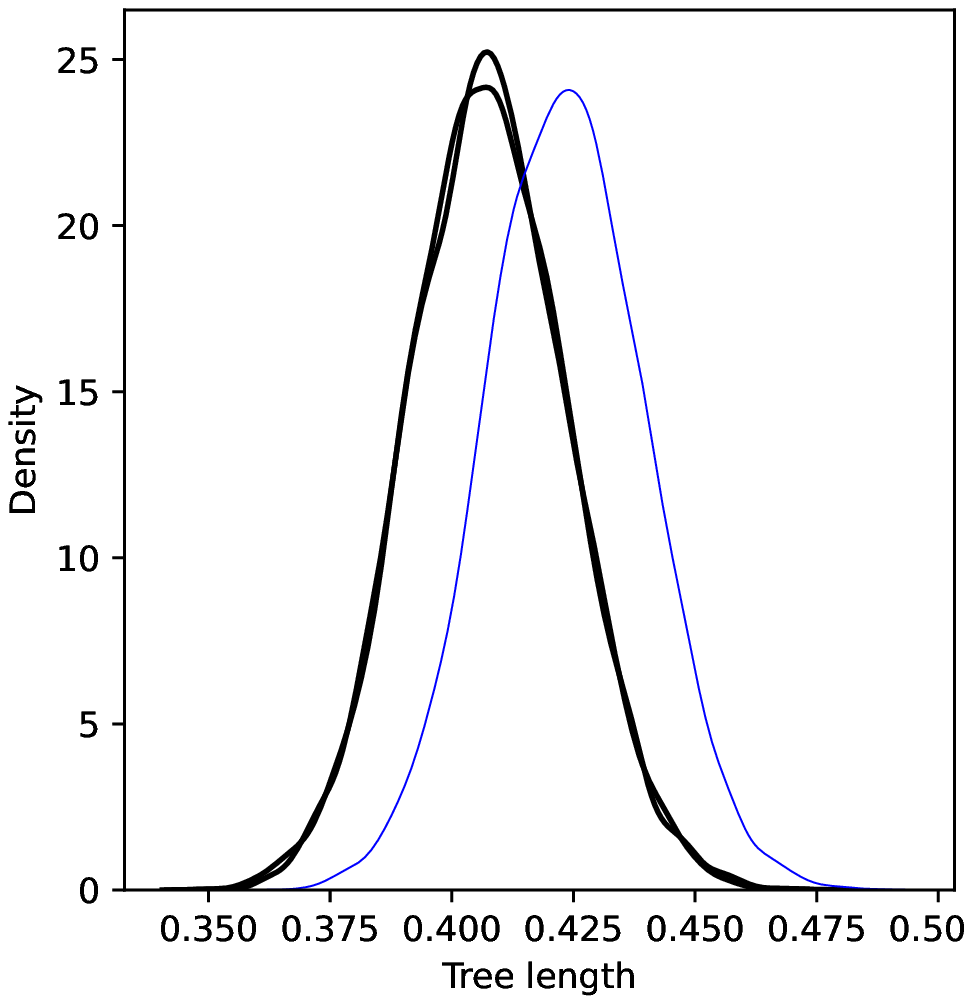}%
\caption{MCMC with wrapped Normal proposals compared to a MrBayes golden run.
a) Mean split lengths.
Leaf edges in blue circles, internal edges in red diamonds.
b) Split frequencies.
c) Tree length distribution.
}
\label{fig:wrap}
\end{center}
\end{figure}

As expected from MCMC, the results of using these proposals is reasonably similar to projecting proposal vectors via with $\phi$. 
Both the split frequencies (ASDSF of $0.004$) and mean lengths are captured well, figure~\ref{fig:wrap}.
The similar results of this method demonstrate the possibility of using alternative proposals.

\section{C: Normal Prior on Locations} \label{app:prior}
Placing a prior on embedding locations, rather than tree samples avoids the Jacobian involved in MCMC for transforming from embeddings to trees.
Here, we use a standard multivariate normal distribution $\mathcal{N}(\vec{0}, \Sigma)$ with covariance given by the $d$ dimensional identity matrix.
As before, Dodonaphy runs for $10^{6}$ iterations including a $10^{4}$ warm-up period using the tuning method presented above.
It starts from the distances on consensus tree from the golden run of MrBayes.

\begin{figure}[htbp]
\begin{center}
\includegraphics[width=0.7\linewidth]{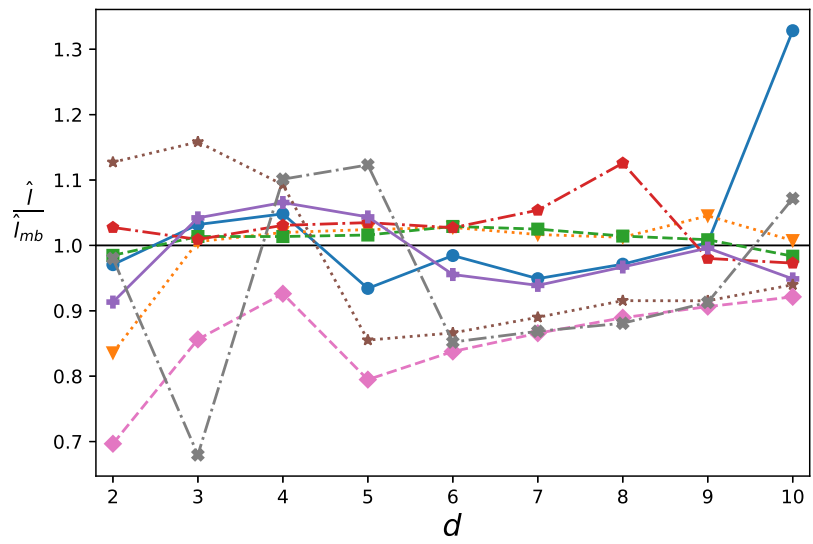}%
\caption{Median tree length estimates compared to MrBayes over varying dimensions.}
\label{fig:prior}
\end{center}
\end{figure}

As expected, the tree length estimates vary from those given by the golden run, merely because they have different priors, figure~\ref{fig:prior}.
More to the point, it demonstrates that there is no clear trend in trees becoming longer with higher dimensions.
This suggests that when a Gamma-Dirichlet prior on trees is used, the trend in increasing tree length (in the main text) derives from the non-trivial Jacobian involved in the MCMC.

\section{D: Closed form of Jacobians}
The Jacobian of the first projection $\phi: \mathbb{R}^{d} \to \mathbb{R}^{d+1}$ is non-square because it attaches an extra dimension to each taxa.
For a taxon with stored location $\bm{x} \in \mathbb{R}^d$, it takes the form 
\be
\frac{\partial{\phi (\bm{x})}}{\partial \bm{x}}  =
\begin{bmatrix}
	\bm{x}_/\hat{\bm{x}} \\
	 I_d
\end{bmatrix}
\in \mathbb{R}^{{d+1}} \times \mathbb{R}^{{d}}
\ee

%To take the determinate we use the generalised Jacobian, which is:
%\be
%\sqrt{|\nabla \phi^{\mathsf T} \nabla \phi|}(\bm{x}) = \sqrt{\frac{||\bm{x}||_2^2}{||\bm{x}||_2^2 + 1} + 1} \leq \sqrt{2}
%\ee

For the second transformation from the hyperboloid to the pairwise distances, the Jacobian the distance between two points $\bm{x}, \bm{y} \in \mathbb{H}^d$ is adapted from~\citet{chowdhary2018improved}:
\be
\frac{\partial d(\bm{x}, \bm{y}, \kappa)}{\partial \bm{x}} = 
\frac{1}{\sqrt{-\kappa((\bm{x} * \bm{y})^2 -1)}} \Big(\frac{\sqrt{||\bm{y}||^2_2+1}}{\sqrt{||\bm{x}||^2_2+1}}\bm{x} - \bm{y}\Big).
\ee
It is clear that the distance between $\bm{x}$ and $\bm{y}$ does not depend on a third point $\bm{z}$:
\be
\frac{\partial d(\bm{x}, \bm{y}, \kappa)}{\partial \bm{z}} = 0.
\ee

\end{document}